\documentclass{article}
\usepackage[T1]{fontenc}
\usepackage[latin9]{inputenc}
\usepackage[a4paper]{geometry}
\usepackage{color}
\usepackage[american]{babel}
\usepackage{array}
\usepackage{float}
\usepackage{amsmath}
\usepackage{amsthm}
\usepackage{amssymb}
\usepackage{graphicx}
\usepackage[authoryear]{natbib}
\usepackage[unicode=true,
 bookmarks=false,
 breaklinks=true,pdfborder={0 0 1},backref=false,colorlinks=false]
 {hyperref}
\hypersetup{
 hidelinks}

\makeatletter

\newcommand{\lyxmathsym}[1]{\ifmmode\begingroup\def\b@ld{bold}
  \text{\ifx\math@version\b@ld\bfseries\fi#1}\endgroup\else#1\fi}

\providecommand{\tabularnewline}{\\}

\theoremstyle{plain}
\newtheorem{thm}{\protect\theoremname}[section]
\theoremstyle{definition}
\newtheorem{defn}[thm]{\protect\definitionname}
\theoremstyle{remark}
\newtheorem{rem}[thm]{\protect\remarkname}
\theoremstyle{plain}
\newtheorem{lem}[thm]{\protect\lemmaname}
\theoremstyle{plain}
\newtheorem{prop}[thm]{\protect\propositionname}
\theoremstyle{plain}
\newtheorem{cor}[thm]{\protect\corollaryname}
\theoremstyle{definition}
\newtheorem{example}[thm]{\protect\examplename}

\makeatother

\addto\captionsamerican{\renewcommand{\corollaryname}{Corollary}}
\addto\captionsamerican{\renewcommand{\definitionname}{Definition}}
\addto\captionsamerican{\renewcommand{\examplename}{Example}}
\addto\captionsamerican{\renewcommand{\lemmaname}{Lemma}}
\addto\captionsamerican{\renewcommand{\propositionname}{Proposition}}
\addto\captionsamerican{\renewcommand{\remarkname}{Remark}}
\addto\captionsamerican{\renewcommand{\theoremname}{Theorem}}
\addto\captionsenglish{\renewcommand{\corollaryname}{Corollary}}
\addto\captionsenglish{\renewcommand{\definitionname}{Definition}}
\addto\captionsenglish{\renewcommand{\examplename}{Example}}
\addto\captionsenglish{\renewcommand{\lemmaname}{Lemma}}
\addto\captionsenglish{\renewcommand{\propositionname}{Proposition}}
\addto\captionsenglish{\renewcommand{\remarkname}{Remark}}
\addto\captionsenglish{\renewcommand{\theoremname}{Theorem}}
\providecommand{\corollaryname}{Corollary}
\providecommand{\definitionname}{Definition}
\providecommand{\examplename}{Example}
\providecommand{\lemmaname}{Lemma}
\providecommand{\propositionname}{Proposition}
\providecommand{\remarkname}{Remark}
\providecommand{\theoremname}{Theorem}

\begin{document}
\title{Profit and loss decomposition in continuous time and approximations}
\author{Gero Junike\thanks{Carl von Ossietzky Universität, Institut für Mathematik, 26129 Oldenburg,
Germany.} \thanks{Corresponding author. ORCID: 0000-0001-8686-2661, E-Mail: gero.junike@uol.de\\ \textcolor{blue}{This article is accepted in the journal Finance and Stochastics.}},
Hauke Stier$^{\ast}$, Marcus Christiansen$^{\ast}$ }
\maketitle

\begin{abstract}
Financial institutions and insurance companies that analyze the evolution and sources of profits and losses often look at risk factors only at discrete reporting dates, ignoring the detailed paths. Continuous-time decompositions avoid this weakness and also make decompositions consistent across different reporting grids. We construct a large class of continuous-time decompositions from a new extended version of Itô's formula and uniquely identify a preferred decomposition from the axioms of exactness, symmetry and normalization. This unique decomposition turns out to be a stochastic limit of recursive Shapley values, but it suffers from a curse of dimensionality as the number of risk factors increases.  We develop an approximation that breaks this curse when the risk factors almost surely have no simultaneous jumps.
\end{abstract}
\begin{verse}


\end{verse}

\section{Introduction}

Profit and loss (P\&L) attribution, also known as change analysis,
has a long history in risk management. P\&L attribution is the process
of analyzing the change between two valuation dates and explaining
the evolution of the P\&L by the movement of the sources (risk factors) 
between the two dates, see \citet{cadoni2014internal}. In other words,
the change in the P\&L over time is \emph{decomposed} in terms of
the different risk factors to explain how each factor \emph{contributes}
to the P\&L. In the literature, there are many ways to obtain a P\&L
attribution. For example, consider a portfolio in EUR consisting of
a long position in the S\&P 500, $Y$ for short. The P\&L of such
a portfolio is driven by two risk factors: $Y$ and the USDEUR exchange
rate, $X$ for short. To decompose the P\&L over one year, we look
for two real numbers $D_{X}$ and $D_{Y}$, such that
\[
X(1)Y(1)-X(0)Y(0)\overset{!}{=}D_{X}+D_{Y}.
\]
The numbers $D_{X}$ and $D_{Y}$ are interpreted as the contribution
of $X$ and $Y$ to the P\&L.  In the literature we can find many desirable properties that a decomposition should possess, see \citet{Shubik1962Incentives}, \cite{friedman1999three} and \citet{shorrocks2013decomposition}
among many others. The authors argue that a decomposition should be \emph{symmetric}, i.e., the contributions
of the risk factors should be independent of the way in which the
risk factors are labeled or ordered. These authors also require that
the sum of all contributions equals the P\&L, such decompositions are called \emph{exact}. Further, \citet{christiansen2022decomposition}
argues that a decomposition should be \emph{normalized}, i.e., if
a risk factor remains constant, its contribution to the P\&L should
be zero. It is also desirable for a decomposition to consider the
full path of each risk factor, i.e., to use all available information,
see \citet{mai2022} and \cite{flaig2024empirical}.

\subsection*{Common decomposition principles}

A common method for creating decompositions is to sequentially update
the risk factors one by one while `freezing' all other risk factors.
This idea dates back at least to \citet{oaxaca1973male} and \citet{blinder1973wage},
who developed a \emph{sequential updating} (SU) decomposition technique
in a one-period setting. The SU decomposition works as follows when we
update the risk factor $X$ first:
\[
D_{X}=X(1)Y(0)-X(0)Y(0),\quad D_{Y}=X(1)Y(1)-X(1)Y(0).
\]
Alternatively, one may update $Y$ first to obtain
\[
D_{X}=X(1)Y(1)-X(0)Y(1),\quad D_{Y}=X(0)Y(1)-X(0)Y(0).
\]
Each SU decomposition is exact, but if there are $d$ risk factors,
there are $d!$ different updating orders and therefore $d!$ different
SU decompositions. \citet{cadoni2014internal} call the one-period SU
decomposition \emph{waterfall} and apply it to P\&L attribution.
See \citet{fortin2011} for an overview on how the SU decomposition
is used in various fields of economics. 

The SU decomposition can also be defined in a multi-period setting
by dividing the time horizon into sub-intervals and applying the SU
decomposition recursively on each sub-interval. \citet{jetses2022general}
and \citet{christiansen2022decomposition} analyzed the limit of the SU decomposition when the mesh size of the
time grid converges to zero. In the limit, the decomposition takes
the whole path into account and the limiting SU decomposition is called
the \emph{infinitesimal sequential updating} (ISU) decomposition.
The ISU decomposition is independent of any time grid, which is helpful
``to prevent inconsistencies when using conflicting sub-intervals
for different purposes'', see \citet[p. 2]{flaig2024empirical}.

The \emph{averaged sequential updating }(ASU) decomposition, also
known as the \emph{Shapley value}, is simply the arithmetic average
of the $d!$ possible SU decompositions. It has many desirable properties,
in particular: it is exact and symmetric. \citet{Shapley1953Game}
introduces the ASU decomposition for cooperative games. \citet{Shubik1962Incentives}
defines the ASU decomposition for cost functions. \citet{Sprumont1998Ordinal}
and \citet{friedman1999three} provide an axiomatization of the ASU
decomposition for cost functions. \citet{jetses2022general}
define the \emph{infinitesimal averaged sequential updating} (IASU)
decomposition as the average of the $d!$ possible ISU decompositions.

\subsection*{Main contributions}

In this paper, we start directly in a time-continuous setting: If
the portfolio is a twice differentiable function of the risk factors
and the risk factors have continuous paths, Itô's formula provides a
natural additive decomposition of the portfolio. Our main contributions
are as follows: In order to treat risk factors with jumps, we prove
an expanded version of Itô's formula and use it to define a large
class of reasonable decompositions, which we call \emph{Itô decompositions}
and which include all $d!$ ISU and the IASU decompositions as special
cases. We prove that there is a unique Itô decomposition (up to indistinguishability) that satisfies
the three axioms of exactness, symmetry and normalization. We show
that it is indistinguishable from the IASU decomposition. We further show that
the IASU decomposition can be interpreted as the limiting case of
the ASU decomposition: compared to \citet{jetses2022general}, who
assume that the covariations between the risk factors are zero, we
use much weaker assumptions to prove the convergence of the SU/ ASU
decompositions to the ISU/ IASU decompositions.

In summary, we propose to use the IASU decomposition to obtain a P\&L
attribution because it considers the whole paths of the risk factors and satisfies the axioms of exactness, symmetry and
normalization. However, in practical applications, the IASU decomposition
has two drawbacks: a) similar to the ASU decomposition, it suffers
from the curse of dimensionality; b) the IASU decomposition is defined
by stochastic integrals, which somehow must be approximated in practice. Naively
approximating these integrals can lead to decompositions that are
no longer exact. As another important contribution of this
paper, we show that the IASU decomposition does not suffer from the
curse of dimensionality if the risk factors do not have simultaneous
jumps. In this case, the IASU decomposition is indistinguishable from the average of two (suitably selected) ISU decompositions. To avoid point b), we suggest approximating ISU/ IASU by SU/ ASU.

Up to now, most practitioners have applied an arbitrary SU decomposition
in a one-period setting to obtain an annual P\&L attribution, see \citet{cadoni2014internal}.
Working with real market data, \cite{flaig2024empirical} empirically
show that the SU decomposition depends significantly on the order
or labeling of the risk factors, and that some SU decompositions may
even ignore relevant risk factors, which may ``lead to wrong trading
and hedging decisions'', see \citet[p. 2]{flaig2024empirical}. 

Our theoretical analysis suggests using the average of only two SU decompositions\footnote{Define one SU decomposition in any order, e.g., alphabetically ascending,
and the other SU decomposition by the reverse order, e.g., alphabetically
descending, see Theorem \ref{thm:IASU_nsj} for details.} with a sufficiently fine time grid to obtain a P\&L attribution, since such a decomposition is arbitrarily close
to the IASU decomposition when the risk factors do not have simultaneous
jumps. Thus, our analysis is highly relevant for practitioners:
we recommend computing two SU decompositions instead of one and using
a finer grid than just annual data to obtain a decomposition that
is much closer to the IASU decomposition than a single SU decomposition. While the choice of the decomposition (the average of two SU decompositions) is theoretically justified, we have only numerical experiments available to estimate the time grid and we recommend using monthly or weekly data.

\subsection*{Further literature review}

Is there any other way to break the curse of dimensionality? \citet{christiansen2022decomposition}
proves that the ISU decomposition is symmetric if it is stable with
respect to small perturbations in the empirical observation of the
risk factors. In Appendix \ref{subsec:Stability}, we show that the
ISU decomposition of a simple product of two correlated Brownian motions
is not stable. This shows that stability is a rather strong assumption.

There are other decomposition principles as well: There
is the so-called \emph{one-at-a-time} (OAT) decomposition, which is
also known as \emph{bump and reset}, see \citet{cadoni2014internal}.
The OAT decomposition is closely related to the SU decomposition.
It is symmetric, but generally not exact.  \citet{frei2020new} analyzes the limit of the OAT decomposition when the mesh size of the
time grid converges to zero.

There are also completely different approaches: \citet{fischer2004decomposition}
uses a conditional expectations approach.  
 \citet{rosen2010risk} use the Hoeffding method for a decomposition
of credit risk portfolios. \citet{frei2020new} uses the Euler principle
for risk attribution. 
\citet{ramlau1991distribution} and \citet{norberg1999theory} decompose surplus in life insurance by heuristic integral representations, where the integrators are interpreted as the driving forces of change and determine the attribution. A similar idea is used in 
 \citet{schilling2020decomposing} based on the martingale representation theorem.

\subsection*{Contents}

In Section \ref{sec:Notation}, we establish some notation. In Section
\ref{sec:Family-of-Ito}, we develop an extended version of Itô's formula
and introduce the family of Itô decompositions. We show that the IASU
decomposition is the only exact and symmetric Itô decomposition, and
we break the curse of dimensionality of the IASU decomposition in
Theorem \ref{thm:IASU_nsj}. In Section \ref{sec:SU,-OAT-and}, we
prove that the IASU decomposition can be approximated by the ASU decomposition.
In Section \ref{sec:Applications}, we provide some numerical applications.
Section \ref{sec:Conclusions} concludes.

\section{\label{sec:Notation}Notation}

Let $\big(\Omega,\mathcal{F},\mathbb{F}=(\mathcal{F}_{t})_{t\geq0},P\big)$
be a filtered probability space satisfying the usual conditions, i.e., $\mathcal{F}_{0}$ contains all null sets and $\mathbb{F}$ is right-continuous. Let
$\mathcal{X}$ be the set of all $\mathbb{F}$-semimartingales. 
A so-called \emph{risk basis} or \emph{information
basis} is a $d$-dimensional semimartingale $X\in\mathcal{X}^d$, and its $d$ components are denoted as \emph{risk factors}
or \emph{sources of risk}. For a stopping time $s$, we define the stopped semimartingale by
$X^{s}=\left(X_{1}^{s},...,X_{d}^{s}\right)$. Equality of random variables is understood in the almost sure sense and equality of stochastic processes is understood up to indistinguishability. Let $\mathcal{C}_{2}$
be the set of twice differentiable functions from $\mathbb{R}^{d}$
to $\mathbb{R}$. For $f\in\mathcal{C}_{2}$ and $i,j=1,...,d$, we
write $f_{i}$ and $f_{ij}$ for the partial derivatives $\partial_{i}f$
and $\partial_{i}\partial_{j}f$. We call a map $F:\mathcal{X}^{d}\to\mathcal{X}$ \emph{non-anticipative}
if for any stopping time $s$ it holds that
\begin{equation}
F\left(X^{s}\right)(t)=F(X)(\min(t,s)),\quad t\geq0.\label{eq:FX^s}
\end{equation}
Such a non-anticipative mapping depends only on the information
up to time $t$, i.e., on $X^{t}$. By $\mathcal{M}$ we denote some
sub-space of all non-anticipative mappings. By $\mathcal{M}(\mathcal{C}_{2})$
we denote the space of functionals $F:\mathcal{X}^{d}\to\mathcal{X}$
such that $F(X)=\left(f(X(t))\right)_{t\geq0}$, $X\in\mathcal{X}^{d}$,
for some $f\in\mathcal{C}_{2}$, which are clearly non-anticipative.
By $\sigma_{d}$ we denote the set of all $d!$ permutations of $\{1,...,d\}$.
Let $id\in\sigma_{d}$ be the identity. In a slight abuse of notation,
we define for $\pi\in\sigma_{d}$
\[
\pi(x_{1},...,x_{d})=\left(x_{\pi(1)},...,x_{\pi(d)}\right),\quad x\in\mathbb{R}^{d},
\]
and
\begin{align*}
\pi(X_{1},...,X_{d}) & =\left(X_{\pi(1)},...,X_{\pi(d)}\right),\quad X\in\mathcal{X}^{d}.
\end{align*}
For two one-dimensional semimartingales $Z$ and $Y$ and a càglàd
process $H$, we denote by $\int_{0}^{t}H(s)dZ(s):=\int_{0+}^{t}H(s)dZ(s)$
the stochastic integral. In particular $\int_{0}^{0}H(s)dZ(s)=0$
by convention. We further set $Z(0-)=0$, 
\[
Z(t-)=\lim_{\varepsilon\searrow0}Z(t-\varepsilon),\quad t>0,\quad\Delta Z(t)=Z(t)-Z(t-),\quad t\geq0,
\]
\[
[Z,Y]=ZY-Z(0)Y(0)-\int_{0}^{\cdot}Z(u-)dY-\int_{0}^{\cdot}Y(u-)dZ
\]
and
\[
[Z,Y]^{c}=[Z,Y]-\sum_{0<s\leq\,\cdot}\Delta Z(s)\Delta Y(s).
\]
By ``$\overset{p}{\rightarrow}$'' we denote the convergence in probability of a sequence of random variables.
For $A\subset\{1,...,d\}$ we define the projection 
\begin{align*}
p_{A}:\mathbb{R}^{d} & \to\mathbb{R}^{d}\\
x & \mapsto\big(x_{1}1_{A}(1),...,x_{d}1_{A}(d)\big),
\end{align*}
where the function $1_{A}(h)$ is one if $h\in A$ and zero otherwise.

\section{\label{sec:Family-of-Ito}Family of Itô decompositions}

Similar to \citet{shorrocks2013decomposition} and \citet{christiansen2022decomposition},
we define a \emph{decomposition} as follows:
\begin{defn}
A map
\begin{align*}
\delta:\mathcal{M}\times\mathcal{X}^{d} & \to\mathcal{X}^{d}\\
(F,X) & \mapsto(\delta_{1}(F,X),...,\delta_{d}(F,X))
\end{align*}
is called a \emph{decomposition}. 
\end{defn}

We interpret $\delta_{i}(F,X)(t)$ as the contribution of $X_{i}$
to the profit and loss $F(X)(t)-F(X)(0)$ in $(0,t]$. We recall the
following three axioms from the literature:
\begin{description}
\item [{i)}] A decomposition is called \emph{exact} if for all $F\in\mathcal{M}$
and all $X\in\mathcal{X}^{d}$ the following equation holds:
\[
F(X)-F(X)(0)=\delta_{1}(F,X)+...+\delta_{d}(F,X).
\]
\item [{ii)}] A decomposition is called \emph{symmetric} if for all $\pi\in\sigma_{d}$,
all $F\in\mathcal{M}$ and all $X\in\mathcal{X}^{d}$ the following
implication holds:
\[
F(X)=F(\pi(X)) \quad \Longrightarrow\quad\delta_{i}(F,X)=\delta_{\pi^{-1}(i)}(F,\pi(X)).
\]
\item [{iii)}] A decomposition is called \emph{normalized} if for all $0\leq r<s<\infty$,
all $i=1,...,d$, all $F\in\mathcal{M}$ and all $X\in\mathcal{X}^{d}$
the following implication holds:
\begin{align*}
&X_{i}\text{ is indistinguishable from a constant process on }(r,s]\\
&\Longrightarrow\delta_{i}(F,X)\text{ is indistinguishable from a constant process on }(r,s].
\end{align*}

\end{description}
Axiom i) ensures that a decomposition is able to fully explain the P\&L, see \citet{shorrocks2013decomposition} and \citet{christiansen2022decomposition}.
Axiom ii) considers symmetric maps $F$ and states that if $F$ does
not depend on the order or labeling of the risk factors, then the
decomposition shall also be independent of the order or labeling of the risk
factors. The symmetry axiom is motivated by the fact that $\delta_{i}(F,X)$
represents the contribution of $X_{i}$ and that the term $\delta_{\pi^{-1}(i)}(F,\pi(X))$
also describes the contribution of 
\[
\pi(X)_{\pi^{-1}(i)}=(X_{\pi(1)},...,X_{\pi(d)})_{\pi^{-1}(i)}=X_{i}.
\]
The symmetry axiom has already been mentioned in similar form in \citet{friedman1999three}
and \citet{shorrocks2013decomposition}. If the risk factor $X_{i}$
remains constant during some time interval $(r,s]$, it does not contribute
to $F(X)(s)-F(X)(r)$, so the contribution of $X_{i}$ in $(r,s]$ should also be zero. This is exactly reflected by the normalization axiom,
taken from \citet{christiansen2022decomposition}. 

Next, we indicate how Itô's formula helps to define decomposition
principles: Let $f:\mathbb{R}^{d}\to\mathbb{R}$ be twice continuously
differentiable. For $i,j=1,...,d$ let
\begin{align}
I_{i} & :=\int_{0}^{\cdot}f_{i}(X(s-))dX_{i}(s),\quad I_{ij}:=\int_{0}^{\cdot}f_{ij}(X(s-))d[X_{i},X_{j}]^{c}(s),\label{eq:Ii}\\
S & :=\sum_{0<s\leq\cdot}\bigg\{ f(X(s))-f(X(s-))-\sum_{i=1}^{d}f_{i}(X(s-))\Delta X_{i}(s)\bigg\}.\label{eq:S}
\end{align}
Itô's formula states that for $t\geq0$ it holds for any semimartingale
$X$ that 
\begin{align}
f(X(t))-f(X(0))= & \sum_{i=1}^{d}I_{i}(t)+\frac{1}{2}\sum_{i=1}^{d}I_{ii}(t)+\frac{1}{2}\sum_{i,j=1 \atop i \neq j}^{d}I_{ij}(t)+S(t).\label{eq:Ito_con}
\end{align}
Assume that $X$ has continuous paths without interaction effects,
i.e., $S=0$ and $I_{ij}=0$, $i\neq j$. Then, Eq.~(\ref{eq:Ito_con}) provides
a natural way to additively decompose the P\&L $f(X(t))-f(X(0))$:
 
By the normalization axiom, $I_{i}$ and $I_{ii}$ should be assigned to $\delta_{i}$, which is interpreted as the contribution of $X_i$. To see this, assume that some $\delta_{j}$
would depend on $I_{i}$ or $I_{ii}$ for $i\neq j$. Assume that
$X_{j}$ is constant everywhere. According to the normalization axiom,
we would then have $\delta_{j}=0$. So, $\delta_{j}$ must not depend
on $I_{i}$ or on $I_{ii}$.

However, how to handle the interaction effects $I_{ij}$, $i\neq j$ and the jump part $S$ is not so obvious. Therefore, in Proposition
\ref{prop:(Generalized-Itos-formula).} we provide an extended version
of Itô's formula. Based on Proposition \ref{prop:(Generalized-Itos-formula).},
we define the large family of\emph{ Itô decompositions} in Definition
\ref{def:ito_decomposition} and show in Section \ref{sec:Family-of-Ito}
that the family of Itô decompositions contains many well-known decomposition
principles as special cases. Within the family of Itô decompositions,
we will identify a single decomposition that satisfies the axioms
of exactness, symmetry and normalization. For $A\subseteq\{1,\dots,d\}$, $i\in\{1,...,d\}$ and $s>0$ define
\begin{align*}
Y_{i}^A(s):= & f\bigg(X(s-)+p_{A}\big(\Delta X(s)\big)\bigg)-f\bigg(X(s-)+p_{A\backslash\{i\}}\big(\Delta X(s)\big)\bigg) -f_{i}\big(X(s-)\big)\Delta X_{i}(s)
\end{align*}
and
\[
S_{i}^{A}(X):=\sum_{0<s\leq\cdot}Y^A_{i}(s).
\]
For $\pi\in\sigma_{d}$ define
\begin{equation}
S_{i}^{\pi}(X):=S_{i}^{\{j\,|\pi(j)\leq\pi(i)\}}(X).\label{eq:Spi_i}
\end{equation}
To obtain $S_{i}^{\pi}(X)$, all time points $s$ where $X_{i}$ jumps
are considered. All risk factors
except $X_{i}$ are fixed at $s$ or $s-$, depending on the choice of $\pi$, and only $X_{i}$ varies
between $s-$ and $s$.
\begin{lem}
\label{lem:S^A_abs_con}Let $i\in\{1,...,d\}$, $X\in\mathcal{X}^d$
and $A\subset\{1,...,d\}$. If $i\in A$, then $S_{i}^{A}(X)$ is
a semimartingale with a.s.~paths of finite variation on compacts.
\end{lem}

\begin{proof}
Fix $X\in\mathcal{X}^d$. Let $N$ be a null set such that $u\mapsto|X_{i}(u)(\omega)|$, $i=1,...,d$,
is càdlàg for $\omega\in\Omega\setminus N$ and 
\begin{equation}
\sum_{h,j=1}^{d}\sum_{0<s\leq t}|\Delta X_{h}(s)(\omega)\Delta X_{j}(s)(\omega)|<\infty,\quad\omega\in\Omega\setminus N,\quad t\geq0.\label{eq:appyN2}
\end{equation}
Such $N$ exists since $X$ is a semimartingale. Let $\omega\in\Omega\setminus N$
and $t\geq0$. Let $\mathcal{M}_{\omega}\subset\mathbb{R}^{d}$ be
the closure of the set $\{X(u)(\omega),\,u\in[0,t]\}$, which is compact.
The function $f$ and its derivatives are continuous and reach a maximum
and minimum on the convex hull of $\mathcal{M}_{\omega}$, which is
compact by Carathéodory's theorem, see \citet[Sec.~2.3]{gruenbaum2013convex}.
Hence, $f$ and its derivatives are bounded on the convex hull of
$\mathcal{M}_{\omega}$. Let $s\in(0,t]$. Let us develop $f$ around
$X(s-)(\omega)$ using a Taylor expansion. We have that 
\begin{align*}
  f\bigg(X(s-)(\omega)+p_{A}\big(\Delta X(s)(\omega)\big)\bigg)
  =f\big(X(s-)(\omega)\big)+\sum_{h\in A}f_{h}\big(X(s-)(\omega)\big)\Delta X_{h}(s)(\omega)+R(\omega),
\end{align*}
where $R(\omega)$ is the remainder of the Taylor expansion, i.e.,
for some $\theta(\omega)\in[0,1]$ it holds that 
\[
R(\omega)=\frac{1}{2}\sum_{h,j\in A}f_{hj}\bigg(X(s-)(\omega)+\theta(\omega)p_{A}\big(\Delta X(s)(\omega)\big)\bigg)\Delta X_{h}(s)(\omega)\Delta X_{j}(s)(\omega).
\]
The term $f\left(X(s-)(\omega)+p_{A\setminus\{i\}}\big(\Delta X(s)(\omega)\big)\right)$
can be treated similarly. Since $i\in A$, it holds for some $C(\omega)>0$,
which does not depend on $s$ or $\theta(\omega)$, that 
\begin{align*}
Y^A_{i}(s)\leq C(\omega)\sum_{h,j\in A}|\Delta X_{h}(s)(\omega)\Delta X_{j}(s)(\omega)|.
\end{align*}
It follows by Inequality (\ref{eq:appyN2}) that 
\begin{equation}
\sum_{0<s\leq t}|Y^A_{i}(s)(\omega)|<\infty,\quad\omega\in\Omega\setminus N.\label{eq:sumYLem}
\end{equation}
Since $t$ was arbitrarily chosen, Inequality (\ref{eq:sumYLem}) implies
that $u\mapsto S_{i}^{A}(X)(u)(\omega)$, $\omega\in\Omega\setminus N$, is càdlàg and of finite variation on compacts.
Therefore, $S_{i}^{A}(X)$ is a semimartingale. 
\end{proof}

\begin{prop}
\label{prop:(Generalized-Itos-formula).}
Let $\pi\in\sigma_{d}$ and $f\in\mathcal{C}_{2}$ and $X\in\mathcal{X}^d$. For all $t\geq0$
it holds that
\begin{align*}
f(X(t))-f(X(0))= & \sum_{i=1}^{d}\bigg\{ I_{i}(t)+\frac{1}{2}I_{ii}(t)+\frac{1}{2}\sum_{\underset{i\neq j}{j=1}}^{d}I_{ij}(t)+S_{i}^{\pi}(t)\bigg\},
\end{align*}
where $I_{i}$ and $I_{ij}$ are defined in Eq.~(\ref{eq:Ii}) and $S_{i}^{\pi}$ is defined in Eq.~(\ref{eq:Spi_i}). 
\end{prop}

\begin{proof}
Since the series telescopes, we have that
\begin{align*}
 & f(X(s))-f(X(s-))\\
&=  \sum_{i=1}^{d}f\bigg(X(s-)+p_{\{j\,|\pi(j)\leq\pi(i)\}}\big(\Delta X(s)\big)\bigg)-f\bigg(X(s-)+p_{\{j\,|\pi(j)<\pi(i)\}}\big(\Delta X(s)\big)\bigg).
\end{align*}
By Inequality~(\ref{eq:sumYLem}),  it holds for any $t\geq0$ that
\begin{equation}
\sum_{i=1}^{d}S_{i}^{\pi}(X)(t)=\sum_{0<s\leq t}\sum_{i=1}^{d}Y^{\{j\,|\pi(j)\leq\pi(i)\}}_{i}(s)=S(t),\label{eq:sumSi_S}
\end{equation}
where $S$ is defined in Eq.~(\ref{eq:S}). The claim follows then by the classical Itô's formula. 
\end{proof}

Proposition \ref{prop:(Generalized-Itos-formula).} generalizes the
classical Itô's formula, because for any $\pi\in\sigma_{d}$ it holds
that $\sum_{i=1}^{d}S_{i}^{\pi}(X)=S$, see Eq.~(\ref{eq:sumSi_S}).

\begin{defn}
\label{def:ito_decomposition}Let $\lambda_{ij}\in[0,1]$ for $i,j=1,...,d$.
Let $\mu_{\pi}\in[0,1]$ for $\pi\in\sigma_{d}$. The decomposition
\begin{align*}
\delta^{\text{Itô}}:\mathcal{M}(\mathcal{C}_{2})\times\mathcal{X}^{d} & \to\mathcal{X}^{d},\quad(F,X)\mapsto(\delta_{1}^{\text{Itô}}(F,X),...,\delta_{d}^{\text{Itô}}(F,X)),
\end{align*}
where
\[
\delta_{i}^{\text{Itô}}(F,X)=I_{i}+\frac{1}{2}I_{ii}+\sum_{\underset{j\neq i}{j=1}}^{d}\lambda_{ij}I_{ij}+\sum_{\pi\in\sigma_{d}}\mu_{\pi}S_{i}^{\pi}(X),\quad i=1,...,d,
\]
is called \emph{Itô decomposition with parameters $(\lambda_{ij})_{i,j=1,...,d}$
and $(\mu_{\pi})_{\pi\in\sigma_{d}}$.}
\end{defn}
The definition of the Itô decomposition is motivated by Proposition \ref{prop:(Generalized-Itos-formula).} and the normalization axiom:
Below Eq.~(\ref{eq:Ito_con}) we already argued that $I_{i}$ and
$I_{ii}$ should be attributed to $X_{i}$ in order to satisfy the
normalization axiom. If parts of the interaction effect $I_{ij}$
were assigned to the contribution of $X_{h}$ for $h\notin\{i,j\}$,
the decomposition would no longer be normalized. Therefore, only the risk factors $X_{i}$ and $X_{j}$ are assigned shares $\lambda_{ij}$ and $\lambda_{ji}$ of the interaction effect $I_{ij}$.

Note that $S_{i}^{\pi}(X)$ contains only jumps in the $i$-th component.
If $S_{i}^{\pi}(X)$ were assigned to the contribution of some $X_{j}$,
$j\neq i$, the normalization axiom would be violated if $X_{j}$ is
constant. Therefore, $S_{i}^{\pi}$ should be assigned
to the contribution of $X_{i}$. Since there are $d!$ different ways
to decompose the jumps without violating neither the normalization
axiom nor the exactness axiom, we propose to assign to $X_{i}$ a weighted
average of all $S_{i}^{\pi}(X)$, $\pi\in\sigma_{d}$.
\begin{rem}
\label{rem:special_case_it}Since each Itô decomposition is linear
in the first argument $F$, a portfolio can be decomposed by decomposing
each individual instrument.
\end{rem}

We recall some special members of the family of Itô decompositions,
namely the IASU and the $d!$ different ISU decompositions,
which were introduced in \citet{jetses2022general}. All Itô decompositions
are normalized. We will prove that the IASU decomposition is the only
Itô decomposition that is exact and symmetric. We will also see that the
ISU decomposition is closely related to the IASU decomposition and
that the IASU decomposition is the limiting case of the well-known
ASU decomposition (also known as Shapley value), which is defined
over a discrete time grid in Section \ref{sec:SU,-OAT-and}.
\begin{defn}
\label{def:IASU}The \emph{IASU (infinitesimal averaged updating)
decomposition} $\delta^{\text{IASU}}:\mathcal{M}(\mathcal{C}_{2})\times\mathcal{X}^{d}\to\mathcal{X}^{d}$
is defined by
\[
\delta_{i}^{\text{IASU}}(F,X)=I_{i}+\frac{1}{2}\sum_{j=1}^{d}I_{ij}+\frac{1}{d!}\sum_{\pi\in\sigma_{d}}S_{i}^{\pi}(X),\quad i=1,...,d.
\]
\end{defn}

\begin{rem}\label{RemarkSimplificationIASU}
\label{rem:FIASU_2^d} The Itô decompositions are overparameterised: based on   
Eq.~(\ref{eq: ito decomposition xi}) in Lemma \ref{lem: parameterisation ito}, we can represent the IASU decomposition as
\[
\delta_{i}^{IASU}(F,X)=I_{i}+\frac{1}{2}\sum_{j=1}^{d}I_{ij}+\sum_{\substack{A\subseteq\{1,\dots,d\}\\
i\in A
}
}S_{i}^{A}(X)\xi_{i,A},
\]
where
\begin{equation}
\xi_{i,A}:=\sum_{\substack{\pi\in\sigma_{d}\\
\{j|\pi(j)\leq\pi(i)\}=A
}
}\frac{1}{d!}=\frac{(|A|-1)!(d-|A|)!}{d!}.\label{eq:xi}
\end{equation}
Hence, the computational cost to obtain $\delta_{i}^{\text{IASU}}$
can be reduced from $\mathcal{O}(d!)$ to $\mathcal{O}(2^{d-1})$ for $d\rightarrow\infty$. 
\end{rem}

\begin{defn}
\label{def:ISU}Let $\pi\in\sigma_{d}$. The \emph{ISU (infinitesimal
updating) decomposition $\delta^{\text{ISU},\pi}:\mathcal{M}(\mathcal{C}_{2})\times\mathcal{X}^{d}\to\mathcal{X}^{d}$
with updating order $\pi$} is defined by
\[
\delta_{i}^{\text{ISU},\pi}(F,X)=I_{i}+\frac{1}{2}I_{ii}+\sum_{\underset{\pi(j)<\pi(i)}{j=1}}^{d}I_{ij}+S_{i}^{\pi}(X),\quad i=1,...,d.
\]
\end{defn}

\begin{thm}
\label{thm:IASU}Every Itô decomposition that is symmetric and exact
is indistinguishable from the IASU decomposition. The IASU decomposition is
related to the ISU decomposition by
\begin{equation}
\delta_{i}^{\text{IASU}}(F,X)=\frac{1}{d!}\sum_{\pi\in\sigma_{d}}\delta_{i}^{\text{ISU},\pi}(F,X),\quad i=1,...,d,\quad X\in\mathcal{X}^d, \quad F\in\mathcal{M}(\mathcal{C}_2).\label{eq:IASU_AvISU}
\end{equation}
\end{thm}

\begin{proof}
The proof of Theorem \ref{thm:IASU} can be found in Appendix \ref{sec:Proof-for-IASU}.
\end{proof}

The next theorem shows that the curse of dimensionality
of the IASU decomposition can be broken if there are no simultaneous
jumps.
\begin{thm}
\label{thm:IASU_nsj}Let $X\in\mathcal{X}^d$ and $F\in\mathcal{M}(\mathcal{C}_2)$. If $\Delta X_{h}\Delta X_{j}=0$ for all $h,j \in \{1,\ldots, d\}$ with  $h\neq j$, 
then
\begin{align}
\delta_{i}^{\text{IASU}}(F,X)= & \frac{1}{2}\big(\delta_{i}^{\text{ISU},\pi}(F,X)+\delta_{i}^{\text{ISU},\pi^{\prime}}(F,X)\big),\quad i=1,...,d,\label{eq:IASU_2ISU}
\end{align}
for any  $\pi\in\sigma_{d}$  and $\pi^{\prime}=d+1-\pi$. 
\end{thm}

\begin{proof}
Let $0<s<\infty$. In the case of $\Delta X_{i}(s)=0$, we have that
\begin{align*}
 & f\left(X(s-)+p_{\{j\,|\pi(j)\leq\pi(i)\}}\big(\Delta X(s)\big)\right)-f\left(X(s-)+p_{\{j\,|\pi(j)<\pi(i)\}}\big(\Delta X(s)\big)\right)\\
&=  f\left(X(s-)+p_{\{j\,|\pi(j)<\pi(i)\}}\big(\Delta X(s)\big)\right)-f\left(X(s-)+p_{\{j\,|\pi(j)<\pi(i)\}}\big(\Delta X(s)\big)\right)\\
&=  0.
\end{align*}
In the case of $\Delta X_{i}(s)\neq0$, it holds that $X_{j}(s)=X_{j}(s-)$
for all $j\neq i$, and hence, 
\begin{align*}
 & f\left(X(s-)+p_{\{j\,|\pi(j)\leq\pi(i)\}}\big(\Delta X(s)\big)\right)-f\left(X(s-)+p_{\{j\,|\pi(j)<\pi(i)\}}\big(\Delta X(s)\big)\right)\\
 & =f\left(X(s)\right)-f\left(X(s-)\right).
\end{align*}
Hence, for $\pi\in\sigma_{d}$ and $i=1,...,d$ it holds that
\begin{align}
\delta_{i}^{\text{ISU},\pi}= & I_{i}+\frac{1}{2}I_{ii}+\sum_{\underset{\pi(j)<\pi(i)}{j=1}}^{d}I_{ij}+\sum_{\underset{\Delta X_{i}(s)\neq0}{0<s\leq\cdot}}\big\{ f\left(X(s)\right)-f\left(X(s-)\right)-f_{i}\big(X(s-)\big)\Delta X_{i}(s)\big\}.\label{eq:nosJ}
\end{align}
Due to Eqs.~(\ref{eq:IASU_AvISU}, \ref{eq:sum a_ij}), we have that
\begin{align}
\delta_{i}^{\text{IASU}}(F,X)= & I_{i}+\frac{1}{2}\sum_{j=1}^{d}I_{ij}+\sum_{\underset{\Delta X_{i}(s)\neq0}{0<s\leq\cdot}}\left\{ f\left(X(s)\right)-f\left(X(s-)\right)-f_{i}\big(X(s-)\big)\Delta X_{i}(s)\right\} .\label{eq:IASU_nosJ}
\end{align}
Let $\delta_{i}^{\text{ISU},\pi}$ be the ISU decomposition with updating
order $\pi\in\sigma_{d}$ and define $\pi^{\prime}(i)=d+1-\pi(i)$, $i=1,...,d$.
Note that 
\begin{equation}
\sum_{\underset{\pi(j)<\pi(i)}{j=1}}^{d}+\sum_{\underset{\pi\prime(j)<\pi\prime(i)}{j=1}}^{d}=\sum_{\underset{\pi(j)<\pi(i)}{j=1}}^{d}+\sum_{\underset{\pi(j)>\pi(i)}{j=1}}^{d}=\sum_{\underset{i\neq j}{j=1}}^{d}.\label{eq:sums_pi}
\end{equation}
Eqs.~(\ref{eq:nosJ}, \ref{eq:IASU_nosJ}, \ref{eq:sums_pi}) imply
Eq.~(\ref{eq:IASU_2ISU}).
\end{proof}

\begin{rem}
Theorem \ref{thm:IASU_nsj} can be generalized to the case where some but
not all risk factors have simultaneous jumps. For example, suppose $d=3$ and
$\Delta X_{1}\Delta X_{j}=0$, $j\in\{2,3\}$ but possibly $\Delta X_{2}\Delta X_{3}\neq0$.
It is then easy to see that Eq.~(\ref{eq:IASU_2ISU}) still holds. Or, if $d=4$
and $\Delta X_{1}\Delta X_{j}=0$, $j\in\{2,3,4\}$, the IASU decomposition
can be written as a weighted average of four ISU decompositions instead of eight ISU decompositions, which would be necessary if all risk factors had simultaneous jumps.
\end{rem}

\begin{cor}
\label{cor:no_interacion}Let $X\in\mathcal{X}^d$ and $F\in\mathcal{M}(\mathcal{C}_2)$. If $[X_{h},X_{j}]=0$ for all $h,j \in \{1, \ldots,d\}$ with  $h\neq j$, then
\begin{align*}
\delta_{i}^{\text{IASU}}(F,X)= & \delta_{i}^{\text{ISU},\pi}(F,X),\quad i=1,...,d,
\end{align*}
where $\pi\in\sigma_{d}$ is arbitrary. 
\end{cor}

\begin{proof}
The assumption $[X_{i},X_{j}]=0$ for \emph{$i\neq j$} implies $\Delta X_{i}\Delta X_{j}=\Delta[X_{i},X_{j}]=0$.
Therefore, $S_{i}^{\pi_{1}}=S_{i}^{\pi_{2}}$, $\pi_{1},\pi_{2}\in\sigma_{d}$,
see the proof of Theorem \ref{thm:IASU_nsj}. The assertion follows directly
from the Definitions \ref{def:IASU} and \ref{def:ISU}.
\end{proof}
\begin{example}
\label{exa:IASU_finite_var}How does the IASU decomposition treat
simultaneous jumps? Let $d=2$ and assume that $X=(X_{1},X_{2})$
is a pure-jump semimartingale of finite variation. Then the IASU decomposition
is given by 
\begin{align*}
\delta_{1}^{\text{IASU}}(F,X)=\frac{1}{2} & \sum_{0<s\leq\cdot}\bigg\{\big\{ f\big(X(s),X_{2}(s-)\big)-f\big(X(s-)\big)\big\}+\big\{f\big(X(s)\big)-f\big(X_{1}(s-),X_{2}(s)\big)\big\}\bigg\},\\
\delta_{2}^{\text{IASU}}(F,X)=\frac{1}{2} & \sum_{0<s\leq\cdot}\bigg\{ \big\{f\big(X(s)\big)-f\big(X_{1}(s),X_{2}(s-)\big)\big\}+\big\{f\big(X_{1}(s-),X_{2}(s)\big)-f\big(X(s-)\big)\big\}\bigg\}.
\end{align*}
The latter formulas are averages of sequential updates from time point
$s-$ to time point $s$. 
\end{example}

\begin{example}
We decompose the P\&L of the portfolio $P=X_{1}X_{2}$ of a foreign
stock, where $X_{1}$ is the foreign exchange rate and $X_{2}$ is
the stock price in the foreign currency. The instantaneous P\&L of
the foreign stock in domestic currency is given by 
\[
dP(t)=X_{1}(t-)dX_{2}(t)+X_{2}(t-)dX_{1}(t)+d[X_{1},X_{2}](t),
\]
i.e., it can be decomposed into the variation of the exchange rate,
variation of the stock price and interaction effects, compare with
\citet{mai2022}. The IASU decomposition equally distributes the interaction effect between $\delta_{1}^{\text{IASU}}$
and $\delta_{2}^{\text{IASU}}$. To see this, observe that
\begin{align*}
\delta_{1}^{\text{IASU}}(F,X) & =\int_{0}^{\cdot}X_{2}(s-)dX_{1}(s)+\frac{1}{2}[X_{1},X_{2}]^{c}\\
 & \quad+\frac{1}{2}\sum_{0<s\leq\cdot}\bigg\{\big\{ X_{1}(s)X_{2}(s-)-X_{1}(s-)X_{2}(s-)\big\}\\
 & \quad+\big\{X_{1}(s)X_{2}(s)-X_{1}(s-)X_{2}(s)\big\}-2X_{2}(s-)\big(X_{1}(s)-X_{1}(s-)\big)\bigg\}\\
 & =\int_{0}^{\cdot}X_{2}(s-)dX_{1}(s)+\frac{1}{2}[X_{1},X_{2}],
\end{align*}
where $F(X)=X_{1}X_{2}$. For $\delta_{2}^{\text{IASU}}$ the reasoning is similar.
\end{example}

\section{\label{sec:SU,-OAT-and}SU and ASU decompositions and their limits}

The time-dynamic SU and ASU decompositions are defined on discrete time grids,
see \citet{jetses2022general} and \citet{christiansen2022decomposition}.
A light introduction to the SU decomposition can be found in \citet{cadoni2014internal}. In this section, we recall the
definitions of these decompositions and we provide sufficient conditions
such that the SU and the ASU decompositions converge to the ISU and
IASU decompositions, respectively, as the mesh size of the discrete
time grid converges to zero. We recall the following definition from
\citet[p. 64]{protter2005stochastic}.
\begin{defn}
An infinite sequence of finite stopping times $0=s_{0}<s_{1}<s_{2}<...$
such that $\sup_{k}s_{k}=\infty$ a.s.~is called an \emph{unbounded random
partition}. A sequence $(\gamma_n)_{n \in \mathbb{N}}$ of unbounded random partitions $\gamma_n =\{0=s_{0}^{n}<s_{1}^{n}<s_{2}^{n}<...\}$  is said to \emph{tend to the identity} if $\sup_{k}|s_{k+1}^{n}-s_{k}^{n}|\to0$
a.s. for $n\to\infty$. 
\end{defn}

\begin{defn}
\label{def:SU}Let $\gamma=\{0=s_{0}<s_{1}<\dots\}$ be an unbounded
random partition. The \emph{SU} \emph{(sequential updating)} \emph{decomposition
$\delta^{\text{\text{SU},\ensuremath{\pi},\ensuremath{\gamma}}}:\mathcal{M}\times\mathcal{X}^{d}\to\mathcal{X}^{d}$}
\emph{with updating order} $\pi\in\sigma_{d}$ is defined by
\begin{align}
\delta_{i}^{\text{SU},\pi,\gamma}(F,X)=  \sum_{l=0}^{\infty}\bigg\{& F\left(X^{s_{l}}+p_{\{j\,|\pi(j)\leq\pi(i)\}}\left(X^{s_{l+1}}-X^{s_{l}}\right)\right)\nonumber \\
 & -F\left(X^{s_{l}}+p_{\{j\,|\pi(j)<\pi(i)\}}\left(X^{s_{l+1}}-X^{s_{l}}\right)\right)\bigg\}.\label{eq:SU_gen}
\end{align}
\end{defn}

In words, divide the time horizon $[0,t]$ into finitely many sub-intervals, and to obtain the contribution of $X_{i}$, fix all risk factors at the
beginning $s_{l}$ or the end $s_{l+1}$ of each sub-interval (depending
on the updating order $\pi$) and allow  only $X_{i}$ to vary between
$s_{l}$ and $s_{l+1}$.

\begin{prop}
The decomposition $\delta^{\text{\text{SU},\ensuremath{\pi},\ensuremath{\gamma}}}:\mathcal{M}\times\mathcal{X}^{d}\to\mathcal{X}^{d}$ is well defined by  Eq.~\eqref{eq:SU_gen} and exact. The sum in Eq.~\eqref{eq:SU_gen} evaluated at $t\in[0,\infty)$ is a.s. finite.
\end{prop}

\begin{proof}
Let $X\in\mathcal{X}^d$,  $F\in\mathcal{M}$,  $\pi\in\sigma_{d}$,  $n\in\mathbb{N}$, and $t\geq0$. By $x\land y$ we denote the minimum of two real numbers $x$ and
$y$. Using Eq.~(\ref{eq:FX^s})
twice, we get
\begin{align}
\delta_{i}^{\text{SU},\pi,\gamma}(F,X)(t\land s_{n})=\sum_{l=0}^{\infty}\bigg\{ & F\left(X^{s_{l}\land s_{n}}+p_{\{j\,|\pi(j)\leq\pi(i)\}}\left(X^{s_{l+1}\land s_{n}}-X^{s_{l}\land s_{n}}\right)\right)(t)\nonumber \\
 & -F\left(X^{s_{l}\land s_{n}}+p_{\{j\,|\pi(j)<\pi(i)\}}\left(X^{s_{l+1}\land s_{n}}-X^{s_{l}\land s_{n}}\right)\right)(t)\bigg\}\label{eq:deltaSU1}\\
=\sum_{l=0}^{n-1}\bigg\{ & F\left(X^{s_{l}}+p_{\{j\,|\pi(j)\leq\pi(i)\}}\left(X^{s_{l+1}}-X^{s_{l}}\right)\right)(t\land s_{n})\nonumber \\
 & -F\left(X^{s_{l}}+p_{\{j\,|\pi(j)<\pi(i)\}}\left(X^{s_{l+1}}-X^{s_{l}}\right)\right)(t\land s_{n})\bigg\}\label{eq:deltaSU2}
\end{align}
since all addends with $l\geq n$ on the right hand-side of Eq.~(\ref{eq:deltaSU1})
are equal to zero. By Eq.~(\ref{eq:deltaSU2}), for each $n$,
the process $\delta_{i}^{\text{SU},\pi,\gamma}(F,X)$ stopped at $s_{n}$
is a finite sum of semimartingales and hence a semimartingale. By
\citet[Part II, Sec. 2]{protter2005stochastic} and since $s_{n}\to\infty$
a.s. for $n\to\infty$, the process $\delta_{i}^{\text{SU},\pi,\gamma}(F,X)$ is a semimartingale
and the decomposition $\delta^{\text{\text{SU},\ensuremath{\pi},\ensuremath{\gamma}}}$
is therefore well defined. The fact that $s_{n}\to\infty$ a.s. implies that the sum in Eq.~\eqref{eq:SU_gen} evaluated at $t$ is a.s. finite. We show exactness: 
Let $x\in\mathbb{R}^{d}$.
Since
\[
p_{\{j\,|\pi(j)\leq\pi(\pi^{-1}(d))\}}(x)=x\quad\text{and}\quad p_{\{j\,|\pi(j)<\pi(\pi^{-1}(1))\}}(x)=0,
\]
we have for any $t\in[0,\infty)$ and $n\in\mathbb{N}$ by Eq.~(\ref{eq:deltaSU2}) that
\begin{align*}
 & \sum_{i=1}^{d}\delta_{i}^{\text{SU},\pi,\gamma}(F,X)(t\land s_{n})\\
 & =\sum_{\underset{i\neq\pi^{-1}(d)}{i=1}}^{d}\sum_{l=0}^{n-1}F\left(X^{s_{l}}+p_{\{j\,|\pi(j)\leq\pi(i)\}}\left(X^{s_{l+1}}-X^{s_{l}}\right)\right)(t\land s_{n})+\sum_{l=0}^{n-1}F\left(X^{s_{l+1}}\right)(t\land s_{n})\\
 & \qquad-\sum_{\underset{i\neq\pi^{-1}(1)}{i=1}}^{d}\sum_{l=0}^{n-1}F\left(X^{s_{l}}+p_{\{j\,|\pi(j)<\pi(i)\}}\left(X^{s_{l+1}}-X^{s_{l}}\right)\right)(t\land s_{n})-\sum_{l=0}^{n-1}F\left(X^{s_{l}}\right)(t\land s_{n}).
\end{align*}
For each $i\in\{1,...,d\}\setminus\{\pi^{-1}(d)\}$ there is exactly
one $k\in\{1,...,d\}\setminus\{\pi^{-1}(1)\}$ such that 
\[
p_{\{j\,|\pi(j)\leq\pi(i)\}}(x)=p_{\{j\,|\pi(j)<\pi(k)\}}(x),
\]
since $\pi(k)=\pi(i)+1$ if and only if $k=\pi^{-1}\big(\pi(i)+1\big)$.
Thus, we get 
\begin{align*}
\sum_{i=1}^{d}\delta_{i}^{\text{SU},\pi,\gamma}(F,X)(t\land s_{n}) & =\sum_{l=0}^{n-1}F\left(X^{s_{l+1}}\right)(t\land s_{n})-\sum_{l=0}^{n-1}F\left(X^{s_{l}}\right)(t\land s_{n})\\
 & =F\left(X^{s_{n}}\right)(t\land s_{n})-F\left(X^{s_{0}}\right)(t\land s_{n})\\
 & =F(X)(t\land s_{n})-F(X)(0).
\end{align*}
Since $t$ and $n$ were arbitrary and $s_{n}\to\infty$ a.s., the decomposition $\delta^{\text{\text{SU},\ensuremath{\pi},\ensuremath{\gamma}}}$
is exact. To see the last point, note that two processes with càdlàg paths are indistinguishable if they are modifications.
\end{proof}

\begin{example}
\label{exa:SU_d2}Assume $d=2$. The SU decomposition with respect
to $\gamma$ defines $d!=2$ decompositions, namely $\delta^{\text{SU},id,\gamma}(F,X)$
and $\delta^{\text{SU},\varrho,\gamma}(F,X)$ with $\varrho(1)=2$
and $\varrho(2)=1$, by 
\begin{align*}
\delta_{1}^{\text{SU},id,\gamma}(F,X) & =\sum_{l=0}^{\infty}\left\{ F\left(X_{1}^{s_{l+1}},X_{2}^{s_{l}}\right)-F\left(X_{1}^{s_{l}},X_{2}^{s_{l}}\right)\right\} ,\\
\delta_{2}^{\text{SU},id,\gamma}(F,X) & =\sum_{l=0}^{\infty}\left\{ F\left(X_{1}^{s_{l+1}},X_{2}^{s_{l+1}}\right)-F\left(X_{1}^{s_{l+1}},X_{2}^{s_{l}}\right)\right\} 
\end{align*}
and
\begin{align*}
\delta_{1}^{\text{SU},\varrho,\gamma}(F,X) & =\sum_{l=0}^{\infty}\left\{ F\left(X_{1}^{s_{l+1}},X_{2}^{s_{l+1}}\right)-F\left(X_{1}^{s_{l}},X_{2}^{s_{l+1}}\right)\right\} ,\\
\delta_{2}^{\text{SU},\varrho,\gamma}(F,X) & =\sum_{l=0}^{\infty}\left\{ F\left(X_{1}^{s_{l}},X_{2}^{s_{l+1}}\right)-F\left(X_{1}^{s_{l}},X_{2}^{s_{l}}\right)\right\} .
\end{align*}
\end{example}

\begin{defn}
\label{def:ASU}Let $\gamma=\{0=s_{0}<s_{1}<\dots\}$ be an unbounded
random partition. The \emph{ASU (averaged sequential updating)} \emph{decomposition
$\delta^{\text{\text{ASU},\ensuremath{\gamma}}}:\mathcal{M}\times\mathcal{X}^{d}\to\mathcal{X}^{d}$}
is defined by 
\[
\delta_{i}^{\text{ASU},\gamma}(F,X)=\frac{1}{d!}\sum_{\pi\in\sigma_{d}}\delta_{i}^{\text{SU},\pi,\gamma}(F,X),\quad i=1,...,d.
\]
\end{defn}

\begin{rem}
\label{rem:ASU_2^d}As in \citet{shorrocks2013decomposition}, we
observe that
\[
\delta_{i}^{ASU,\gamma}(F,X)=\frac{1}{d!}\sum_{\pi\in\sigma_{d}}\delta_{i}^{\text{SU},\pi,\gamma}(F,X)=\sum_{\substack{A\subseteq\{1,\dots,d\}\\
i\in A
}
}\delta_{i}^{\text{SU},A,\gamma}(F,X)\xi_{i,A}
\]
for $\xi_{i,A}$ defined in Eq.~(\ref{eq:xi}) and 
\begin{align*}
\delta_{i}^{\text{SU},A,\gamma}(F,X) & :=\sum_{l=0}^{\infty}\bigg\{ F\left(X^{s_{l}}+p_{A}\left(X^{s_{l+1}}-X^{s_{l}}\right)\right)-F\left(X^{s_{l}}+p_{A\backslash\{i\}}\left(X^{s_{l+1}}-X^{s_{l}}\right)\right)\bigg\}.
\end{align*}
Thereby, the
computational cost to obtain $\delta_{i}^{ASU,\gamma}$ can be reduced
from $\mathcal{O}(d!)$ to $\mathcal{O}(2^{d-1})$.
\end{rem}

\begin{thm}
\label{thm:convergence}Let $\pi\in\sigma_{d}$ and $(\gamma_{n})_{n\in\mathbb{N}}$
be a sequence of unbounded random partitions tending to the identity. Let $F\in\mathcal{M}(\mathcal{C}_{2})$, 
$X\in\mathcal{X}^{d}$,  $t\geq0$ and $i\in\{1,...,d\}$. Then it holds for $n\to\infty$ that
\begin{align*}
\delta_{i}^{\text{SU},\pi,\gamma_{n}}(F,X)(t) & \overset{p}{\to}\delta_{i}^{\text{ISU},\pi}(F,X)(t),\\
\delta_{i}^{\text{ASU},\gamma_{n}}(F,X)(t) & \overset{p}{\to}\delta_{i}^{\text{IASU}}(F,X)(t).
\end{align*}
\end{thm}

\begin{proof}
The proof of Theorem \ref{thm:convergence} can be found in Appendix
\ref{subsec:Proof_convergence}.
\end{proof}
The next example shows that the assumption $F\in\mathcal{M}(\mathcal{C}_{2})$
in Theorem \ref{thm:convergence} is important to ensure convergence. 
\begin{example}
\label{exa:f not smooth}Let $Z$ be a stochastic process with independent
increments and $Z_{0}=0$. Jumps of $Z$ shall only occur at fixed
times $J=\{2-l^{-1},\,l\in\mathbb{N}\},$ and for each $l\in\mathbb{N}$,
the process jumps by $\pm l^{-1}$ with equal probability for upward
and downward movements. The process $Z$ stays constant between jumps.
Then, $Z$ is a semimartingale, see \citet{vcerny2021pure}. Let 
\[
f(x_{1},x_{2})=|x_{1}-x_{2}|,
\]
so $f\notin\mathcal{C}_{2}$. Let $(\gamma_{n})_{n\in\mathbb{N}}$ 
be a deterministic sequence of unbounded partitions $\gamma_n =\{ 0=s_0^n < s_1^n  < \cdots\}$  tending to the
identity such that $\gamma_n$ contains
the first $n$ smallest elements of $J$ but the intersection with
$(2-n^{-1},2]$ is empty. Assume that $X=(Z,Z)$. Then, for $t\geq2$
it follows that 
\begin{align*}
\sum_{l=0}^{\infty}\left\{ f\big(X_{1}^{s_{l+1}^{n}}(t),X_{2}^{s_{l}^{n}}(t)\big)-f\big(X_{1}^{s_{l}^{n}}(t),X_{2}^{s_{l}^{n}}(t)\big)\right\}  & =\sum_{l=1}^{n}l^{-1},
\end{align*}
which is divergent for $n\to\infty$, so the SU decomposition does
not converge for the map $F(X)(t):=f(X(t))$, $t\geq0$. 
\end{example}
How can the IASU decomposition be computed efficiently in practice?
If we naively approximate the integrals in Definition \ref{def:IASU}
numerically, then we may lose exactness of the decomposition, which
is undesirable in many applications. Theorem \ref{thm:convergence}
suggests using the ASU decomposition as an approximation of the IASU
decomposition. However, this becomes computationally infeasible for moderately large 
$d$, since the computational cost to obtain $\delta_{i}^{ASU,\gamma}$
scales like $\mathcal{O}(2^{d-1})$. The next corollary provides an elegant solution when there are no simultaneous jumps. 

\begin{defn}
Let $\gamma=\{0=s_{0}<s_{1}<\dots\}$ be an unbounded random partition.
The \emph{2SU (average of two sequential updating)} \emph{decomposition
$\delta^{\text{\text{2SU},\ensuremath{\pi},\ensuremath{\gamma}}}:\mathcal{M}\times\mathcal{X}^{d}\to\mathcal{X}^{d}$}
\emph{with updating order $\pi\in\sigma_{d}$} is defined by 
\[
\delta_{i}^{\text{\text{2SU},\ensuremath{\pi},\ensuremath{\gamma}}}(F,X)=\frac{1}{2}\left(\delta_{i}^{\text{SU},\pi,\gamma}(F,X)+\delta_{i}^{\text{SU},\pi^{\prime},\gamma}(F,X)\right),\quad i=1,...,d,
\]
where $\pi^{\prime}=d+1-\pi$.
\end{defn}

\begin{cor}
\label{cor:2SU}Let $\pi\in\sigma_{d}$ and $(\gamma_{n})_{n\in\mathbb{N}}$
be a sequence of unbounded random partitions tending to the identity. Let $F\in\mathcal{M}(\mathcal{C}_{2})$,
$X\in\mathcal{X}^{d}$, $i\in\{1,...,d\}$ and $t\geq0$. 
\begin{description}
	\item [{i)}] If $\Delta X_{h}\Delta X_{j}=0$ for all $h,j \in \{1,\ldots, d\}$ with  $h\neq j$, then
	\begin{align*}
		\delta_{i}^{\text{2SU},\pi,\gamma_{n}}(F,X)(t) & \overset{p}{\to}\delta_{i}^{\text{IASU}}(F,X)(t),\quad n\to\infty.
	\end{align*}
	\item [{ii)}] If $[X_{h},X_{j}]=0$  for all $h,j \in \{1,\ldots, d\}$ with  $h\neq j$, then 
	\[
	\delta_{i}^{\text{SU},\pi,\gamma_{n}}(F,X)(t)\overset{p}{\to}\delta_{i}^{\text{IASU}}(F,X)(t),\quad n\to\infty.
	\]
\end{description}
\end{cor}
\begin{proof}
If $\Delta X_{h}\Delta X_{j}=0$, $h\neq j$, it holds by Theorem \ref{thm:IASU_nsj}
that $\delta_{i}^{\text{IASU}}(F,X)=\frac{1}{2}\big(\delta_{i}^{\text{ISU},\pi}(F,X)+\delta_{i}^{\text{ISU},\pi^{\prime}}(F,X)\big)$,
which is the limit of $\delta_{i}^{2\text{SU},\pi,\gamma_{n}}(F,X)$
by Theorem \ref{thm:convergence}. If $[X_{h},X_{j}]=0$, $h\neq j$,
apply Corollary \ref{cor:no_interacion} and Theorem \ref{thm:convergence}.
\end{proof}
In particular, the 2SU decomposition with arbitrary updating order
$\pi$ is exact and approximates the IASU decomposition when the risk
factors do not have simultaneous jumps. In this case, the computationally
expensive averaging to obtain the ASU decomposition can be omitted
and the computational complexity to approximate $\delta_{i}^{\text{IASU}}$
decreases from $\mathcal{O}(2^{d-1})$ to $\mathcal{O}(1)$. Theorem \ref{thm:convergence}
and Corollary \ref{cor:2SU} are also illustrated in Figure \ref{fig:con_to_IASU}.

\begin{figure}[H]
\[
\begin{array}{ccccc}
 &  & \boxed{\delta_{i}^{\text{ASU},\gamma}}\\
\\
 &  & \Big\downarrow \text{\scriptsize $p$}\\
\\ 
\boxed{\delta_{i}^{\text{2SU},\pi,\gamma}} & \overset{p}{\xrightarrow{\hspace*{0.5cm}}} & \boxed{\delta_{i}^{\text{IASU}}} & \overset{p}{\xleftarrow{\hspace*{0.5cm}}} & \boxed{\delta_{i}^{\text{SU},\pi,\gamma}}\\
\text{ if }\Delta X_{h}\Delta X_{j}=0,\,h\neq j &  &  &  & \text{ if }[X_{h},X_{j}]=0,\,h\neq j
\end{array}
\]

\caption{\label{fig:con_to_IASU}Overview of discrete approximations of the
IASU decomposition.}
\end{figure}

Last, we define the OAT decomposition.  To obtain the contribution of $X_{i}$, all risk factors
are fixed at the origin and only $X_{i}$ is allowed to change from
the beginning of a sub-interval to the end of that sub-interval.
\begin{defn}
\label{def:OAT}Let $\gamma=\{0=s_{0}<s_{1}<\dots\}$ be an unbounded
random partition. The \emph{OAT (one-at-a-time) decomposition $\delta^{\text{OAT,\ensuremath{\gamma}}}:\mathcal{M}\times\mathcal{X}^{d}\to\mathcal{X}^{d}$
}is defined by
\begin{align*}
\delta_{i}^{\text{OAT},\gamma}(F,X)= & \sum_{l=0}^{\infty}\left\{ F\left(X_{1}^{s_{l}},...,X_{i-1}^{s_{l}},X_{i}^{s_{l+1}},X_{i+1}^{s_{l}},...,X_{d}^{s_{l}}\right)-F\left(X^{s_{l}}\right)\right\} ,\quad i=1,...,d.
\end{align*}
\end{defn}

\begin{rem}
\label{rem:IOAT}The OAT decomposition is symmetric but in general not exact. Let $(\gamma_{n})_{n\in\mathbb{N}}$ be a sequence
of unbounded random partitions tending to the identity. For each $i\in\{1,...,d\}$ choose
a permutation $\pi_{i}\in\sigma_{d}$ such that $\pi_{i}(i)=1$. Then
 $\delta_{i}^{\text{OAT},\gamma_{n}}$ is indistinguishable from
 $\delta_{i}^{\text{SU},\pi_{i},\gamma_{n}}$. If
$F\in\mathcal{M}(\mathcal{C}_{2})$ then it holds by Theorem \ref{thm:convergence} for $t\geq0$
that
\begin{align*}
\delta_{i}^{\text{OAT},\gamma_{n}}(F,X)(t) & \overset{p}{\to}\delta_{i}^{\text{ISU},\pi_{i}}(F,X)(t),\quad i=1,...,d
\end{align*}
for $n\to\infty$. Thus, by Corollary \ref{cor:no_interacion},
the three decompositions principles OAT, SU (with arbitrary order $\pi\in\sigma_{d}$) 
and ASU are asymptotically indistinguishable if there are no interaction effects.
\end{rem}

\section{\label{sec:Applications}Applications}

Investment portfolios of financial institutions or insurance companies
may include instruments such as stocks, plain vanilla or callable
bonds, convertible bonds, inflation-linked bonds, contingent convertible
bonds (CoCos), basket options, foreign exchange options and structured
products. 
These instruments often depend on multiple risk factors such as different
foreign exchange rates, interest rates for different maturities, credit
spreads, inflation rate, some trigger activations for CoCos, multiple
equities and time decay. \citet{cadoni2014internal} also considered
defaults and rating changes as risk factors.

In order to obtain a P\&L attribution of such instruments, we propose
the IASU decomposition because it is exact, symmetric and normalized,
and it takes into account the whole paths of the risk factors, i.e.,
uses all available information. The last point also avoids inconsistencies
when reporting a P\&L attribution for different time grids, e.g.,
on an annual, quarterly, monthly and weekly basis. The IASU decomposition involves
a stochastic integral. To approximate the IASU decomposition, we propose
the ASU or 2SU decomposition with a sufficiently fine time grid, as
such an approximation is always an exact decomposition. The use of the 2SU decomposition is
theoretically justified when the risk factors do not have simultaneous
jumps.

In Section \ref{subsec:Decomposing-a-call}, we provide an exemplary decomposition of
a plain vanilla call option with stochastic interest rates on a foreign
stock. A change in the P\&L of this option can be explained by movements
in the stock, the yield curve, the foreign exchange rate and time
decay. Thus, there are $d=4$ risk factors. We analyze the unexplained
P\&L of the OAT decomposition, the range of the SU and 2SU decompositions
over all possible updating orders $\pi\in\sigma_{d}$ for different
time grids, and the convergence of the ASU decomposition to the IASU
decomposition.

Computing the ASU decomposition to approximate the IASU decomposition
becomes infeasible when the number of risk factors $d$ is moderately
large: For example, a plain vanilla bond paying coupons may depend
on $d$ yield curves. A basket option may depend on $d$ stocks. In
practice, $d=30$ is a common case for basket options, see \citet{grzelak2023efficient}.
In Section \ref{subsec:Decomposing-a-basket}, we decompose a digital
cash-or-nothing basket put option. We illustrate that it is impossible
to obtain the ASU decomposition in reasonable time when $d=30$ and
we show how the 2SU decomposition is able to break the curse of dimensionality.

\subsection{\label{subsec:Decomposing-a-call}Decomposing a call option with
stochastic interest rates}

In this section, we allocate the P\&L of the price of a plain vanilla
European call option with strike $K$ and maturity $T=10$ with stochastic
interest rates and foreign exchange exposure. The stock price $S$
is given by a Black-Scholes model with constant volatility $\sigma_{S}>0$
and with stochastic interest rates $r$. The dynamics under the risk
neutral measure are given by
\[
dS(t)=r(t)S(t)dt+\sigma_{S}S(t)dB_{S}(t)
\]
and
\[
dr(t)=\kappa(\eta-r(t))dt+\sigma_{r}dB_{r}(t)
\]
with constant volatility $\sigma_{r}>0$, long term mean $\eta\in\mathbb{R}$
and speed of mean reversion $\kappa>0$. Under the physical measure,
the stock has drift $\mu_{S}\in\mathbb{R}$ and the foreign exchange
rate $Y$ is assumed to follow a geometric Brownian motion with drift
$\mu_{Y}\in\mathbb{R}$ and volatility $\sigma_{Y}>0$ driven by the Brownian motion $B_Y$. The Brownian
motions are assumed to have correlations
\[
dB_{S}(t)dB_{r}(t)=\rho_{Sr}dt,\quad dB_{S}(t)dB_{Y}(t)=\rho_{SY}dt,\quad\text{and}\quad dB_{Y}(t)dB_{r}(t)=\rho_{Yr}dt.
\]
The time left to maturity is denoted by $\tau(t)=T-t$. The price
of the plain vanilla call option $p_{call}(t)$ at time $t$ is given
by a twice differentiable function $f:\mathbb{R}^{d}\rightarrow\mathbb{R}$,
see \citet{rabinovitch1989pricing}, i.e.,
\[
p_{call}(t)=f\big(S(t),r(t),Y(t),\tau(t)\big)\\
=:F(S,r,Y,\tau)(t),\quad t\geq0,
\]
with
\[
f(s,r,y,\tau)=ys\Phi\big(d_{+}(s,r,\tau)\big)-yKP(r,\tau)\Phi\big(d_{-}(s,r,\tau)\big),
\]
where $\Phi$ denotes the distribution function of a standard normal distribution and
\begin{align*}
	d_\pm(s,r,\tau)&=\frac{\log\left(\frac{s}{KP(r,\tau)}\right)\pm \frac{1}{2}v(\tau)}{\sqrt{v(\tau)}},\\
	v(\tau)&=\sigma_S^2\tau+\sigma_r^2\frac{\tau-2g_\kappa(\tau)+g_{2\kappa}(\tau)}{\kappa^2}-2\rho_{Sr}\sigma_S\sigma_r\frac{\tau-g_\kappa(\tau)}{\kappa},\\
	g_\kappa(\tau)&=\frac{1-e^{-\kappa \tau}}{\kappa}.
\end{align*}
The bond price $P(r,\tau)$ is given by
\[
P(r,\tau)=A(\tau)e^{-g_\kappa(\tau)r},
\]
where
\[
A(\tau)=\exp\bigg(\bigg(\eta+\frac{\sigma_{r}^{2}\lambda}{\kappa}-\frac{\sigma_{r}^{2}}{2\kappa^{2}}\bigg)\left(g_{\kappa}(\tau)-\tau\right)-\frac{1}{\kappa}\bigg(\frac{\sigma_{r}g_{\kappa}(\tau)}{2}\bigg)^2\bigg)
\]
and $\lambda$ denotes the market price of risk. For simplicity, we set the market price
of risk to zero and hence assume that the dynamics of $r$ under the physical and the risk neutral measure
are identical. \citet[Sec. 24.2]{bjork2009arbitrage} describes how to estimate the parameters for $r$
from market data.
We simulate $1000$ paths of the stock, interest rate and foreign
exchange rate under the physical measure over one year. For each path,
we decompose the price of the call option at time $t=1$ with respect to  the $d=4$ risk factors $X:=(S,r,Y,\tau)$. We use
the following parameters:
\[
K=S(0)=100,\quad\mu_{S}=0.05,\quad\sigma_{S}=0.4,\quad Y(0)=1.1,\quad\mu_{Y}=0,\quad\sigma_{Y}=0.05
\]
and
\[
r(0)=0.08,\quad\kappa=0.1,\quad\eta=0.05,\quad\sigma_{r}=0.01,\quad\rho_{Sr}=-0.7,\quad\rho_{SY}=-0.4,\quad\rho_{Yr}=0.7.
\]

By $\Delta F:=F(X)(1)-F(X)(0)$, we denote the P\&L
of the option over one year. Figure \ref{fig:unexp PL} shows the
relative unexplained P\&L of the OAT decomposition, i.e.,
\[
\frac{|\Delta F-\sum_{i=1}^{d}\delta_{i}^{OAT,\gamma}(F,X)(1)|}{|\Delta F|}.
\]
We use as time grids $\gamma$ annual, quarterly, monthly, weekly and daily time steps. As observed
in \citet{flaig2024empirical}, we also see that the unexplained P\&L of the OAT decomposition
is significant for all time grids. 

Figure \ref{fig:range SU} shows the relative range of the $d!$ SU
decompositions for the risk factor $S$, i.e,
\[
\max_{\pi\in\sigma_{d}}\bigg(\frac{\delta_{1}^{SU,\pi,\gamma}(F,X)(1)}{\delta_{1}^{IASU}(F,X)(1)}\bigg)-\min_{\pi\in\sigma_{d}}\bigg(\frac{\delta_{1}^{SU,\pi,\gamma}(F,X)(1)}{\delta_{1}^{IASU}(F,X)(1)}\bigg)
\]
and the relative range of the $\frac{d!}{2}$ 2SU decompositions for the risk factor $S$. The limiting IASU decomposition is approximated by an ASU decomposition with $10,000$ time steps per year. We observe that the range is significant for the SU decompositions  and insignificant for the 2SU decompositions.

The speed of convergence of the ASU to the IASU decomposition is illustrated in Figure \ref{fig:convergence ASU} for the risk factor $S$, i.e., we show the convergence of 
\[
\frac{\delta_{1}^{ASU,\gamma}(F,X)(1)}{\delta_{1}^{IASU}(F,X)(1)}
\]
to one when the partition $\gamma$ tends to the identity. Figures \ref{fig:range SU} and \ref{fig:convergence ASU} look similar for other risk factors.

In further numerical experiments, we calculate  the relative difference
between the ASU decomposition and the 2SU decompositions
\[
\left|\frac{\delta_{i}^{2SU,\pi,\gamma}(F,X)(1)-\delta_{i}^{ASU,\gamma}(F,X)(1)}{\delta_{i}^{IASU,\gamma}(F,X)(1)}\right|
\]
over all risk factors $i\in\{1,...,d\}$, time grids $\gamma$ and updating orders $\pi\in\sigma_d$, and observe values of less than $0.6\%$ in  $95\%$ of the simulations.  
In conclusion, we find that the ASU decomposition and the 2SU decompositions
are strongly dependent on the time grid, but using monthly or weekly
time steps instead of annual time steps significantly reduces the
deviation of the ASU and 2SU decompositions from the IASU decomposition.

\begin{figure}[H]
\centering{}\includegraphics[scale=0.5]{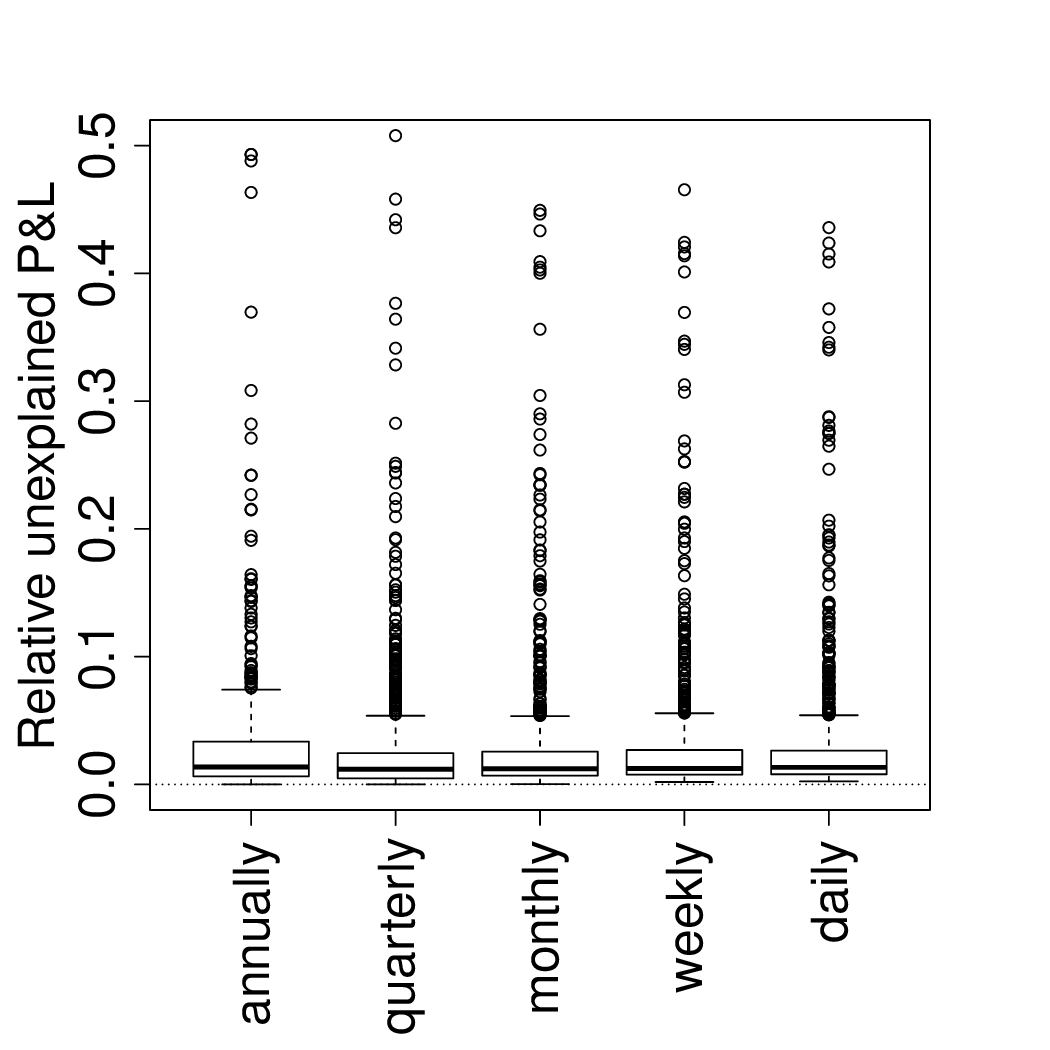}\caption{\label{fig:unexp PL}Relative unexplained P\&L for the OAT decomposition
of a plain vanilla call option in a foreign currency at time $t=1$
for different time grids.}
\end{figure}
\begin{figure}[H]
\centering{}\includegraphics[scale=0.5]{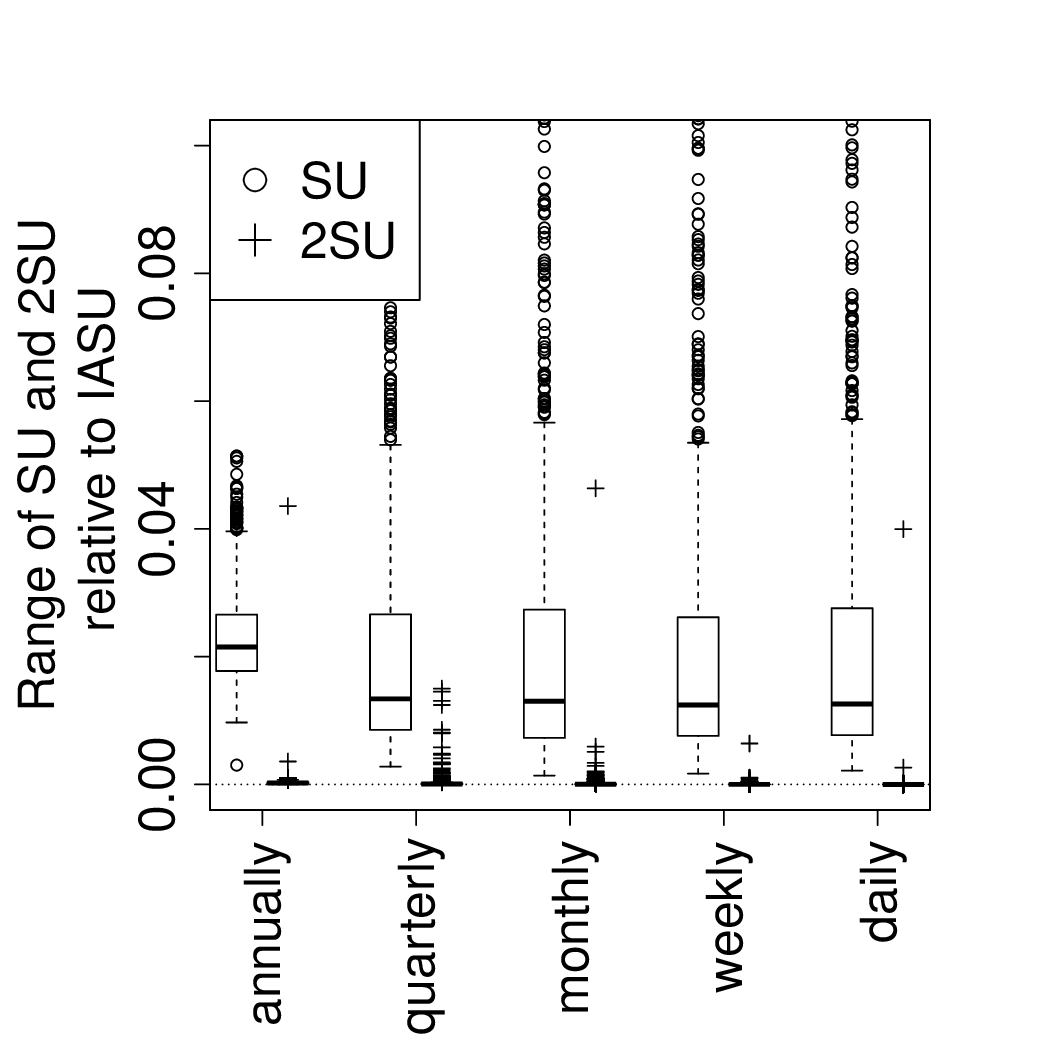}\caption{\label{fig:range SU}Relative range of all SU decompositions
and 2SU decompositions for the risk factor $S$.}
\end{figure}
\begin{figure}[H]
	\centering{}\includegraphics[scale=0.5]{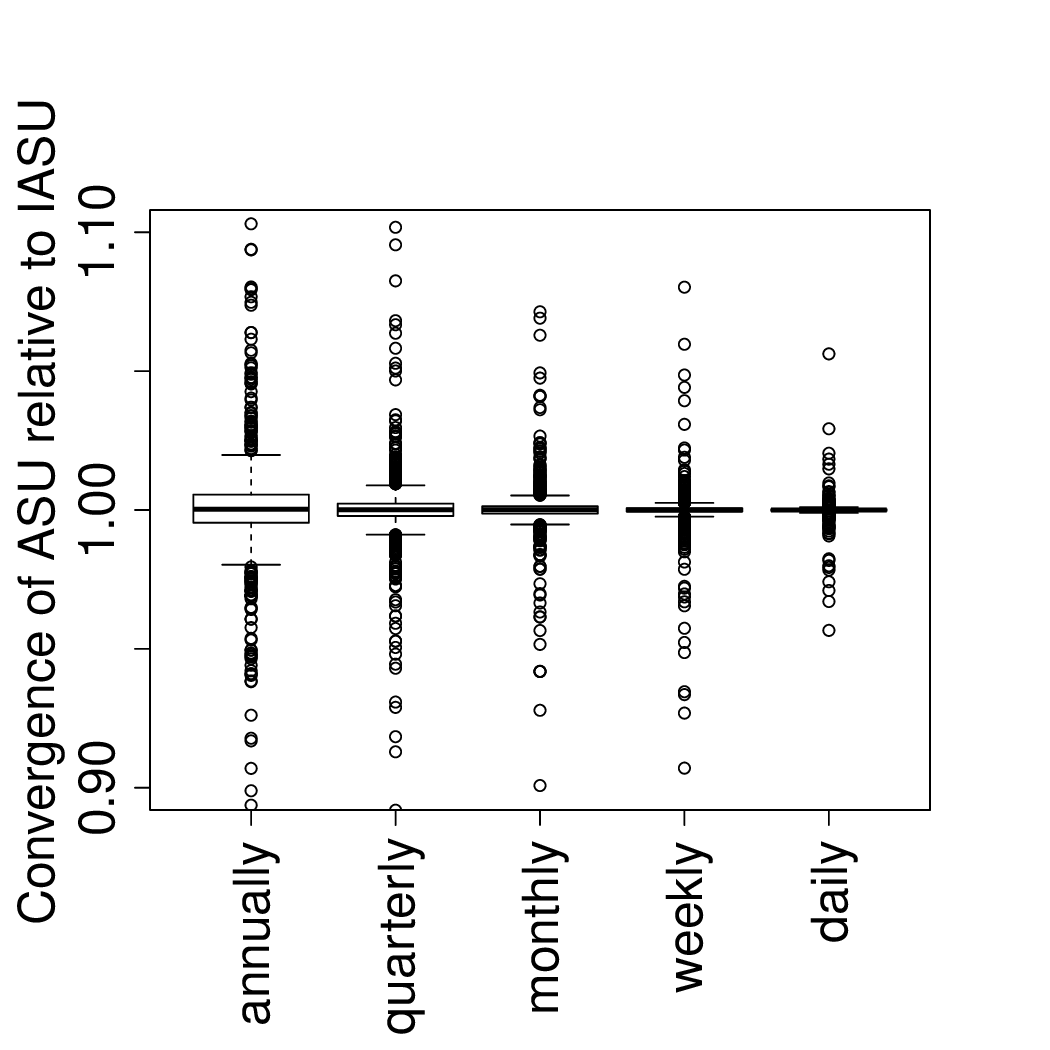}\caption{\label{fig:convergence ASU}Convergence of the ASU decomposition
	to the IASU decomposition for the risk factor $S$.}
\end{figure}

\subsection{\label{subsec:Decomposing-a-basket}Decomposing a basket option}

In this section, we compare the computational cost of obtaining a
one-year P\&L attribution of a basket option using a naive SU decomposition
with annual time grid with the computational cost of obtaining an
ASU and a 2SU decomposition based on a monthly time grid, respectively.
We consider $d$ risk factors: time decay and $d-1$ different
stocks. A digital cash-or-nothing basket put option pays $\$1$ at maturity
$T$ if $S_{1}(T)\leq K,\dots,S_{d-1}(T)\leq K$ and zero otherwise.
The stock prices are given by a Black-Scholes model. We set the interest
rate $r$ to zero. We set $S_{i}(0)=K=100$, $i=1,\dots,d-1$ and
$T=2$. The price of the option at time $t\in[0,T)$ is equal to $\Phi\big(\log(K),\dots,\log(K)\big)$,
where $\Phi$ is the distribution function of a $d-1$ dimensional normal distribution with location
\[
\left(\log\left(S_{1}(t)\right)-\big(r-\frac{1}{2}\sigma^{2}\big)(T-t),\dots,\log\left(S_{d-1}(t)\right)-\big(r-\frac{1}{2}\sigma^{2}\big)(T-t)\right)\in\mathbb{R}^{d-1}
\]
and covariance matrix $\Sigma(T-t)$, where we set $\sigma=0.2$, $\rho=0.5$
and
\[
\Sigma_{ij}=\begin{cases}
\sigma^{2}, & i=j\\
\rho\sigma^{2}, & i\neq j.
\end{cases}
\]
Basket options are often priced using Monte Carlo techniques, see
\citet{glasserman2004monte}. For moderate dimensions, many basket options can also be priced using faster Fourier techniques, see \citet{eberlein2010analysis}
and \citet{junike2024}. We compute $\Phi$ using a simple Monte Carlo simulation implemented in C++ with $100,000$ simulations. The experiments are performed on a laptop with Intel i7-11850H processor and 32 GB RAM.

Table \ref{tab:cpu-time} shows the CPU time needed to obtain $\Phi$ for $d\in \{4,15,30\}$. We measure CPU times by averaging over $100$ runs. Since in some cases the arguments of $\Phi$ to obtain a SU decomposition with a certain update order $\pi$ are the same for different contributions, we need to evaluate $\Phi$ only $dL+1$ times, where $L$ is the number of sub-intervals of $[0,T]$, to obtain the $d$ individual contributions. For example, the 2SU and ASU decompositions with a monthly time grid require $(12d+1)\cdot2$ and $(12d+1)\cdot2^{d-1}$ evaluations of $\Phi$, respectively.

Table \ref{tab:cpu-time} also shows the CPU time to compute the SU, ASU and 2SU decompositions. A naive SU decomposition based on an annual time grid is at most 24 times faster than a 2SU decomposition with a monthly time grid. The computational cost of the 2SU decomposition for each contribution
is dimension independent, except for the longer time required to evaluate $\Phi$. Compared to the ASU decomposition, the 2SU decomposition is $2^{d-2}$ times faster. The ASU decomposition cannot be computed in reasonable time for $d\geq30$. 
\begin{table}[H]
\centering{}
\begin{tabular}{|c|>{\centering\arraybackslash}p{3cm}|c|c|c|}  
\hline
& Number of evaluations of $\Phi$ & $d=4$ & $d=15$ & $d=30$
\tabularnewline
\hline
\hline
Evaluation of $\Phi$ & $1$ & 0.018 sec & 0.15 sec & 0.54 sec \tabularnewline
\hline
SU with annual grid & $d+1$ & (0.09 sec) & (2.4 sec) & (16.7 sec)
\tabularnewline
\hline
2SU with monthly grid & $(12d+1)\cdot 2$ & (1.76 sec) & (54.3 sec) & (390 sec)\tabularnewline
\hline
ASU with monthly grid & $(12d+1)\cdot 2^{(d-1)}$ & (7.06 sec) & (123.6 hours) & (3318.7 years)\tabularnewline
\hline
\end{tabular}
\caption{\label{tab:cpu-time}CPU time to compute the $d$ contributions of the SU, ASU and 2SU decompositions of a basket option over one year using different time grids. The CPU time of $\Phi$ is obtained from a Monte Carlo simulation. The CPU times in brackets are estimated using the CPU time of $\Phi$ and the known complexities of the three decompositions.}
\end{table}

\begin{rem}
	To reduce the computational time, it is possible to compute the $d$ contributions for the SU, 2SU and ASU decompositions in parallel, which would reduce the numerical effort by a factor of $d$. Furthermore, the sums for the SU, 2SU and ASU decompositions can also be parallelized. For example, for the 2SU decomposition we need to perform $2(dL+1)$ function evaluations to obtain all $d$ contributions. If a function evaluation takes $0.54$ sec in $d=30$ dimensions as in the Table \ref{tab:cpu-time}, the computation time for the 2SU decomposition with monthly time grid could be reduced from $390$ sec to about $0.54$ sec using $722$ cores for parallelization.
\end{rem}

\section{\label{sec:Conclusions}Conclusions}

We showed that the IASU decomposition is the only (up to indistinguishability) exact and symmetric decomposition in the family of Itô decompositions, which is a large class of normalized decompositions based on an extended version of Itô's formula. This axiomatic result, together with the fact that the IASU decomposition is grid-independent and considers the full paths of the risk basis, makes it a decomposition of choice from a theoretical perspective. In practice, the calculation of the IASU decomposition comes with two challenges: it involves stochastic integrals that must be approximated, and the computational effort explodes as the number of risk factors increases. 

We have shown that the IASU decomposition can be approximated by the ASU decomposition (which is always exact and symmetric) if we use a sufficiently fine time grid, but the ASU decomposition also suffers from the curse of dimensionality as the number of risk factors increases. For applications where different risk factors may have interactions but almost surely do not have simultaneous jumps, we have shown that the IASU decomposition is indistinguishable from the average of two ISU decompositions, thus breaking the curse of dimensionality. Therefore, from a theoretical point of view, the 2SU decomposition with sufficiently fine time steps is an appropriate approximation of the IASU decomposition.

Based on our own numerical experiments and the empirical analysis of  \cite{flaig2024empirical}, we recommend using monthly or even weekly time steps instead of annual time steps.

The additional computational cost of our two recommendations is moderate, but the theoretical properties of the decomposition are dramatically improved.

\section*{Acknowledgements}

We thank Bernd Buchwald and Andreas Märkert from Hannover Re, Solveig
Flaig from Deutsche Rück,  Julien Hambuckers from University of Liège,  Jan-Frederik Mai from XAIA Investment GmbH, Thorsten Schmidt from Universtität Freiburg and an anonymous referee for very helpful comments that improved this paper.

\appendix

\section{Appendix}

\subsection{Auxiliary results}
\begin{lem}
Let $i,j\in\{1,....,d\}$. Let $\pi,\eta\in\sigma_{d}$ and $x\in\mathbb{R}^{d}$.
Then it holds that
\begin{equation}
\eta^{-1}\bigg(p_{\{j\,|\pi(j)\leq\pi(\eta^{-1}(i))\}}\big(\eta(x)\big)\bigg)=p_{\{j\,|\pi(\eta^{-1}(j))\leq\pi(\eta^{-1}(i))\}}(x).\label{eq:eta p}
\end{equation}
\end{lem}

\begin{proof}
Let $k\in  \big\{j|\pi(j)\leq\pi(\eta^{-1}(i))\big\}$, which is equivalent to 
\begin{align*}
\eta(k)\in & \big\{j|\pi(\eta^{-1}(j))\leq\pi(\eta^{-1}(i))\big\}.
\end{align*}
Since $\big(\eta^{-1}(x)\big)_{\eta(k)}=x_{k}$ and $\big(\eta(x)\big)_{k}=x_{\eta(k)}$,
we obtain that
\begin{align*}
 \bigg(\eta^{-1}\big(p_{\{j\,|\pi(j)\leq\pi(\eta^{-1}(i))\}}\big(\eta(x)\big)\big)\bigg)_{\eta(k)}
 &= \bigg(p_{\{j\,|\pi(j)\leq\pi(\eta^{-1}(i))\}}\big(\eta(x)\big)\bigg)_{k}\\
 &= \big(p_{\{j\,|\pi(\eta^{-1}(j))\leq\pi(\eta^{-1}(i))\}}(x)\big)_{\eta(k)},
\end{align*}
which leads to Eq.~(\ref{eq:eta p}). 
\end{proof}
\begin{lem}
\label{lem: parameterisation ito} Let $\eta\in\sigma_{d}$, $i\in\{1,...,d\}$, $X\in\mathcal{X}^d$,
 $F\in\mathcal{M}(\mathcal{C}_{2})$ and $(\mu_{\pi})_{\pi\in\sigma_{d}}\subset[0,1]$. If $F(\eta(X))=F(X)$, then it
holds that 
\[
\sum_{\pi\in\sigma_{d}}\mu_{\pi}S_{\eta^{-1}(i)}^{\pi}\big(\eta(X)\big)=\sum_{\substack{A\subseteq\{1,\dots,d\}\\
i\in A
}
}S_{i}^{A}(X)\xi_{i,A,\eta}
\]
with

\begin{equation}
\xi_{i,A,\eta}:=\sum_{\substack{\pi\in\sigma_{d}\\
\{j|\pi(\eta^{-1}(j))\leq\pi(\eta^{-1}(i))\}=A
}
}\mu_{\pi}.\label{eq:xiiAeta}
\end{equation}
In particular, for an Itô decomposition $\delta$ with parameters $(\lambda_{ij})_{i,j=1,...,d}$ and $(\mu_{\pi})_{\pi\in\sigma_{d}}$, we have that
\begin{align}
\delta_{i}(F,X) & =I_{i}+\frac{1}{2}I_{ii}+\sum_{\underset{j\neq i}{j=1}}^{d}\lambda_{ij}I_{ij}+\sum_{\substack{A\subseteq\{1,\dots,d\}\\
i\in A
}
}S_{i}^{A}(X)\xi_{i,A,id}.\label{eq: ito decomposition xi}
\end{align}
 
\end{lem}

\begin{proof}
Let $\eta\in\sigma_d$ and $F(X)(t)=f(X(t))$, $t\geq0$, with $F(\eta(X))=F(X)$ for $X\in\mathcal{X}^{d}$.
Let $i\in\{1,...,d\}$. By Eq.~(\ref{eq:eta p}) it holds for $s>0$ that
\begin{align*}
 & f\bigg(\eta\big(X(s-)\big)+p_{\{j\,|\pi(j)\leq\pi(\eta^{-1}(i))\}}\big(\Delta\eta(X)(s)\big)\bigg)\\
 & \overset{\phantom{(\ref{eq:eta p})}}{=}f\bigg(\eta\bigg[X(s-)+\eta^{-1}\big[p_{\{j\,|\pi(j)\leq\pi(\eta^{-1}(i))\}}\big(\eta(\Delta X(s))\big)\big]\bigg]\bigg)\\
 &\overset{(\ref{eq:eta p})}{=} f\bigg(\eta\big[X(s-)+p_{\{j\,|\pi(\eta^{-1}(j))\leq\pi(\eta^{-1}(i))\}}(\Delta X(s))\big]\bigg)\\
 &\overset{\phantom{(\ref{eq:eta p})}}{=} f\bigg(X(s-)+p_{\{j\,|\pi(\eta^{-1}(j))\leq\pi(\eta^{-1}(i))\}}\big(\Delta X(s)\big)\bigg).
\end{align*}
The last equality follows from the symmetry of $f$. Similarly, if we 
replace ``$\leq$'' with ``$<$'', we get that
\[
f\bigg(\eta\big(X(s-)\big)+p_{\{j\,|\pi(j)<\pi(\eta^{-1}(i))\}}\big(\Delta\eta(X)(s)\big)\bigg)=f\bigg(X(s-)+p_{\{j\,|\pi(\eta^{-1}(j))<\pi(\eta^{-1}(i))\}}\big(\Delta(X)(s)\big)\bigg).
\]
Let $\eta\in\sigma_{d}$ and $f\in\mathcal{C}_{2}$. If $f(x)=f(\eta(x))$, $x\in\mathbb{R}^d$,
it is straightforward to see that for $x\in\mathbb{R}^{d}$ it holds
that
\begin{equation}
f_{i}(x)=f_{\eta^{-1}(i)}(\eta(x)),\quad f_{ij}(x)=f_{\eta^{-1}(i)\eta^{-1}(j)}(\eta(x))\quad\text{and}\quad(\eta(x))_{\eta^{-1}(i)}=x_{i}.\label{eq:f pi}
\end{equation}
Therefore it follows that 
\begin{equation}
S_{\eta^{-1}(i)}^{\pi}(\eta(X))=S_{i}^{\pi\circ\eta^{-1}}(X).\label{eq:SiS eta}
\end{equation}
Thus, similarly to \citet{shorrocks2013decomposition}, for any re-ordering $\eta(X)$ of the risk basis $X$ we can conclude
that
\begin{align*}
\sum_{\pi\in\sigma_{d}}\mu_{\pi}S_{\eta^{-1}(i)}^{\pi}\big(\eta(X)\big)\overset{(\ref{eq:SiS eta})}{=} & \sum_{\pi\in\sigma_{d}}\mu_{\pi}S_{i}^{\pi\circ\eta^{-1}}(X)\\
= & \sum_{\substack{A\subseteq\{1,\dots,d\}\\
i\in A
}
}\sum_{\substack{\pi\in\sigma_{d}\\
\{j|\pi(\eta^{-1}(j))\leq\pi(\eta^{-1}(i))\}=A
}
}\mu_{\pi}S_{i}^{\pi\circ\eta^{-1}}(X)\\
= & \sum_{\substack{A\subseteq\{1,\dots,d\}\\
i\in A
}
}S_{i}^{A}(X)\sum_{\substack{\pi\in\sigma_{d}\\
\{j|\pi(\eta^{-1}(j))\leq\pi(\eta^{-1}(i))\}=A
}
}\mu_{\pi}\\
= & \sum_{\substack{A\subseteq\{1,\dots,d\}\\
i\in A
}
}S_{i}^{A}(X)\xi_{i,A,\eta}.
\end{align*}
Eq.~(\ref{eq: ito decomposition xi}) follows directly for $\eta=id.$
\end{proof}
\begin{lem}
\label{lem:xi}Let $\delta$ be an Itô decomposition with parameters
$(\lambda_{ij})_{i,j=1,...,d}$ and $(\mu_{\pi})_{\pi\in\sigma_{d}}$. Let $i\in\{1,\dots,d\}$. If $\delta$ is symmetric and exact, it follows that
\begin{equation}
\xi_{i,A,id}=\xi_{\eta^{-1}(i),\eta^{-1}(A),id}\label{eq: xi symm}
\end{equation}
for any $\eta\in\sigma_{d}$, where $\xi_{i,A,id}$ is defined in
Eq.~(\ref{eq:xiiAeta}) and $\eta(A):=\{\eta(j):j\in A\}$.
Further, for any  $a\in\{1,\dots,d\}$   it holds that 
\begin{equation}
\sum_{j=1}^{d}\sum_{\substack{A\subseteq\{1,\dots,d\}\\
|A|=a,\,j\in A
}
}\xi_{j,A,id}=1.\label{eq:xi add}
\end{equation}
\end{lem}

\begin{proof}
First we show Eq.~(\ref{eq: xi symm}). Let $A\subseteq\{1,\dots,d$\}
with $i\in A$. Let $\pi,\eta\in\sigma_{d}$. Because of
\[
\{j|\pi(\eta^{-1}(j))\leq\pi(\eta^{-1}(i))\}=A\quad \Leftrightarrow \quad \{j|\pi(j)\leq\pi(\eta^{-1}(i))\}=\eta^{-1}(A),
\]
it holds that
\begin{align}
\xi_{i,A,\eta} & =\sum_{\substack{\pi\in\sigma_{d}\\
\{j|\pi(\eta^{-1}(j))\leq\pi(\eta^{-1}(i))\}=A
}
}\mu_{\pi}=\sum_{\substack{\pi\in\sigma_{d}\\
\{j|\pi(j)\leq\pi(\eta^{-1}(i))\}=\eta^{-1}(A)
}
}\mu_{\pi}=\xi_{\eta^{-1}(i),\eta^{-1}(A),id}.\label{eq:xi_eta xi_id}
\end{align}
Now let $f(x)=\prod_{j=1}^{d}x_{j}^{2}$ and $F(X)(t)=f(X(t))$, $t\geq0$, so that 
$F(X)=F(\pi(X))$, $\pi\in\sigma_{d}$. For $B\subseteq\{1,\dots,d\}$
with $i\in B$ and $t\geq0$, let
\[
X_{j}(t)=\begin{cases}
1_{[1,\infty)}(t), & j\in B\\
1_{[0,1)}(t), & j\notin B.
\end{cases}
\]
Then it follows that
\[
f\bigg(X(1-)+p_{A}\big(\Delta X(1)\big)\bigg)=\begin{cases}
1, & A=B\\
0, & A\neq B
\end{cases}
\]
and therefore
\[
S_{i}^{A}(X)(1)=\begin{cases}
1, & A=B\\
0, & A\neq B
\end{cases}
\]
for $A\subseteq\{1,\dots,d\}$ with $i\in A$. For $\eta\in\sigma_{d}$
it follows by Lemma \ref{lem: parameterisation ito} that
\begin{align*}
\delta_{\eta^{-1}(i)}\big(F,\eta(X)\big)(1)= & \sum_{\substack{A\subseteq\{1,\dots,d\}\\
i\in A
}
}S_{i}^{A}(X)(1)\xi_{i,A,\eta}=\xi_{i,B,\eta}.
\end{align*}
Since $\delta$ is symmetric, we have that 
\begin{align*}
\xi_{\eta^{-1}(i),\eta^{-1}(B),id}\overset{(\ref{eq:xi_eta xi_id})}{=}\xi_{i,B,\eta}=\delta_{\eta^{-1}(i)}\big(F,\eta(X)\big)(1) & =\delta_{i}(F,X)(1)=\xi_{i,B,id}.
\end{align*}
Since $B$ was arbitrary, we have just shown Eq.~(\ref{eq: xi symm}).

Now we iteratively show Eq.~(\ref{eq:xi add}). Let $X_{j}(t)=1_{[1,\infty)}(t)$, $t\geq0$,
$j=1,\dots,d$ and let $f^{a}\in\mathcal{C}_{2}$ such that for $a\in\{1,\dots,d\}$
\[
f^{a}(x)=\begin{cases}
1, & \sum_{j=1}^{d}x_{j}=a\\
0, & \sum_{j=1}^{d}x_{j}\in(-\infty,a-1]\cup[a+1,\infty)
\end{cases}
\]
and $f_{i}^{a}(X)=0$ if $\sum_{j=1}^{d}x_{j}\leq a-1$, $i=1,\dots,d$.
Let $F^{a}(X)(t)=f^{a}(X(t))$, $t\geq0$. If $a=d$, then
\[
S_{j}^{A}(X)(1)=\begin{cases}
1, & |A|=a\\
0, & \text{otherwise}
\end{cases}
\]
for $j=1,\dots,d$ and $A\subseteq\{1,\dots,d\}$ with $j\in A$.
By exactness and Lemma \ref{lem: parameterisation ito} it follows
that 
\begin{align}
1 & =F^{a}(X)(1)-F^{a}(X)(0)\nonumber \\
 & =\sum_{j=1}^{d}\delta_{j}(F^{a},X)(1)\nonumber \\
 & =\sum_{j=1}^{d}\sum_{\substack{A\subseteq\{1,\dots,d\}\\
j\in A
}
}S_{j}^{A}(X)(1)\xi_{j,A,id}\nonumber \\
 & =\sum_{j=1}^{d}\sum_{\substack{A\subseteq\{1,\dots,d\}\\
|A|=d,\,j\in A
}
}\xi_{j,A,id}.\label{eq:a=00003Dd}
\end{align}
Now let $a=d-1$, then
\[
S_{j}^{A}(X)(1)=\begin{cases}
1, & |A|=a\\
-1, & |A|=a+1\\
0, & \text{otherwise}
\end{cases}
\] 
for $A\subseteq\{1,\dots,d\}$ with $j\in A$. Again, by exactness
we have that
\begin{align*}
0= & F^{a}(X)(1)-F^{a}(X)(0)\\
= & \sum_{j=1}^{d}\delta_{j}(F^{a},X)\\
= & \sum_{j=1}^{d}\sum_{\substack{A\subseteq\{1,\dots,d\}\\
j\in A
}
}S_{j}^{A}(X)(1)\xi_{j,A,id}\\
= & \sum_{j=1}^{d}\sum_{\substack{A\subseteq\{1,\dots,d\}\\
|A|=d-1,\,j\in A
}
}\xi_{j,A,id}-\sum_{j=1}^{d}\sum_{\substack{A\subseteq\{1,\dots,d\}\\
|A|=d,\,j\in A
}
}\xi_{j,A,id}.
\end{align*}
Using Eq.~(\ref{eq:a=00003Dd}) we obtain that
\[
\sum_{j=1}^{d}\sum_{\substack{A\subseteq\{1,\dots,d\}\\
|A|=d-1,\,j\in A
}
}\xi_{j,A,id}=1.
\]
Iteratively for any $a\in\{1,\dots,d\}$ it follows that
\[
\sum_{j=1}^{d}\sum_{\substack{A\subseteq\{1,\dots,d\}\\
|A|=a,\,j\in A
}
}\xi_{j,A,id}=1.
\]
\end{proof}

\subsection{\label{sec:Proof-for-IASU}Proof of Theorem \ref{thm:IASU}}
\begin{proof}
First we show that the IASU decomposition is exact and symmetric and
satisfies Eq.~(\ref{eq:IASU_AvISU}): By Proposition \ref{prop:(Generalized-Itos-formula).},
it follows that $\delta^{\text{IASU}}$ is an exact Itô decomposition.
Use Eq.~(\ref{eq:f pi}) to see that the IASU decomposition is symmetric.
If $d=1$, Eq.~(\ref{eq:IASU_AvISU}) is trivially true. Assume $d\geq2$.
Fix $i\in\{1,...,d\}$. Note that 
\[
\sum_{\pi\in\sigma_{d}}1_{\{\pi(j)<\pi(i)\}}=\begin{cases}
\frac{d!}{2}, & j\neq i\\
0, & j=i.
\end{cases}
\]
It follows that 
\begin{align}
\frac{1}{d!}\sum_{\pi\in\sigma_{d}}\sum_{\underset{\pi(j)<\pi(i)}{j=1}}^{d}I_{ij} & =\sum_{j=1}^{d}\bigg\{ I_{ij}\frac{1}{d!}\sum_{\pi\in\sigma_{d}}1_{\{\pi(j)<\pi(i)\}}\bigg\}\nonumber \\
 & =\frac{1}{2}I_{i1}+...+\frac{1}{2}I_{i(i-1)}+\frac{1}{2}I_{i(i+1)}+...+\frac{1}{2}I_{id}\nonumber \\
 & =\frac{1}{2}\sum_{j\neq i}I_{ij}.\label{eq:sum a_ij}
\end{align}
Eq.~(\ref{eq:sum a_ij}) implies Eq.~(\ref{eq:IASU_AvISU}).
Now we show that all exact and symmetric Itô decompositions are indistinguishable from the IASU decomposition. Let $\delta$ be a symmetric and exact
Itô decomposition with parameters $(\lambda_{ij})_{i,j=1,...,d}$
and $(\mu_{\pi})_{\pi\in\sigma_{d}}$. Since the  Itô decomposition is over-parameterised, we use the alternative parametrization according to  Eq.~(\ref{eq: ito decomposition xi}).
To prove that $\delta$ is indistinguishable from the IASU decomposition, we
show that $\lambda_{ij}$ and $\xi_{i,A,id}$ are equal to the coefficients
$\frac{1}{2}$ and $\xi_{i,A}$ as defined in Eq.~(\ref{eq:xi}). 

Suppose that $\lambda_{hk}\neq\frac{1}{2}$. Let $X\in\mathcal{X}^{d}$
have continuous paths with $X_{i}=1$, $i\notin\{h,k\}$, and $[X_{h},X_{k}]\neq0$.
Let $F(X)=\prod_{i=1}^{d}X_{i}$. Then $F(X)=F(\pi(X))$ for $\pi\in\sigma_{d}$.
Note that $I_{kh}=I_{hk}$. As $\delta$ is exact, we
have
\[
\sum_{i=1}^{d}\delta_{i}(F,X)=I_{h}+I_{k}+\lambda_{hk}I_{hk}+\lambda_{kh}I_{kh}=F(X)-F(X)(0)=I_{h}+I_{k}+I_{hk},
\]
hence $\lambda_{kh}=1-\lambda_{hk}\neq\lambda_{hk}$. Let $\pi\in\sigma_{d}$
such that $\pi^{-1}(h)=k$. Then, it follows that
\[
\delta_{\pi^{-1}(h)}(F,\pi(X))=\delta_{k}(F,\pi(X))=I_{h}+\lambda_{kh}I_{kh}\neq I_{h}+\lambda_{hk}I_{hk}=\delta_{h}(F,X).
\]
That means that $\delta$ is not symmetric, which is a contradiction to our assumption.  So we necessarily have that  $\lambda_{ij}=\frac{1}{2}$,
$i,j=1,...,d$.

Now let $a\in\{1,\dots,d\}$. For $i,j\in\{1,\dots,d\}$, let $A,B\subseteq\{1,\dots,d\}$
with $|A|=|B|=a$ and $i\in A$, $j\in B$. Then there is a permutation
$\eta\in\sigma_{d}$ such that $\eta^{-1}(A)=B$ and $j=\eta^{-1}(i)$.
By Eq.~(\ref{eq: xi symm}) it follows that

\begin{equation}
\xi_{i,A,id}=\xi_{j,B,id}.\label{eq:xi symm 2sets}
\end{equation}
Let $A_{1},\dots,A_{d}\subseteq\{1,\dots,d\}$ with $j\in A_{j}$
and $|A_{j}|=a$, $j=1,\dots,d$. Since 
\begin{equation}
\left|\big\{ A\subseteq\{1,\dots,d\}:j\in A,|A|=a\big\} \right|=\binom{d-1}{a-1},\label{eq: d-1 over a-1}
\end{equation}
we obtain by Eqs\emph{.~}(\ref{eq:xi add}), (\ref{eq:xi symm 2sets}) and (\ref{eq: d-1 over a-1}) that
\begin{align*}
1= \sum_{j=1}^{d}\sum_{\substack{A\subseteq\{1,\dots,d\}\\
|A|=a,\,j\in A
}
}\xi_{j,A,id} 
=\sum_{j=1}^{d}\binom{d-1}{a-1}\xi_{j,A_{j},id} = d\binom{d-1}{a-1}\xi_{i,A,id} 
\end{align*}
for $A\subseteq\{1,\dots,d\}$ with $i\in A$ and $|A|=a$. Therefore
we can conclude that 
\[
\xi_{i,A,id}=\frac{1}{d\binom{d-1}{|A|-1}}=\frac{(|A|-1)!(d-|A|)!}{d!}.
\]
\end{proof}

\subsection{\label{subsec:Proof_convergence}Proof of Theorem \ref{thm:convergence}}
\begin{proof}
Let $t>0$. Fix some $i\in\{1,...,d\}$ and some permutation $\pi$.
Since $F\in\mathcal{M}(\mathcal{C}_{2})$, by definition there is an $f\in\mathcal{C}_{2}$ such that $F(X)(t)=f(X(t))$, $t\geq0$. We first show that $\delta^{\text{SU},\pi,\gamma_{n}}(F,X)(t)\overset{p}{\to}\delta^{\text{ISU},\pi}(F,X)(t)$ for $n\to\infty$.
 Let $\gamma_{n}=\{0=s_{0}^{n}<s_{1}^{n}<...\}$, $n \in \mathbb{N}$, 
be a sequence of unbounded random partitions tending to the identity.
Let $\alpha>0$ and 
\[
\mathcal{A}_{\alpha}:=\big\{ s\in(0,t]\, :\,\max_{j=1,...,d}|\Delta X_{j}(s)|>\alpha\big\}.
\]
The set $\mathcal{A}_{\alpha}$ contains all time points in $[0,t]$ where
at least one component of a path $u\mapsto X(u)$ has jumps greater
than $\alpha$. The SU decomposition $\delta_{i}^{\text{SU},\pi,\gamma_{n}}$
with respect to $\gamma_{n}$ can be written as 
\begin{align}
\delta_{i}^{\text{SU},\pi,\gamma_{n}}(F,X)  (t)=&\sum_{l\in\mathbb{A}_{\alpha}}\bigg\{ f\left(X^{s_{l}^{n}}(t)+p_{\{j\,|\pi(j)\leq\pi(i)\}}\left(X^{s_{l+1}^{n}}(t)-X^{s_{l}^{n}}(t)\right)\right)\nonumber \\
 & \quad\quad\quad-f\left(X^{s_{l}^{n}}(t)+p_{\{j\,|\pi(j)<\pi(i)\}}\left(X^{s_{l+1}^{n}}(t)-X^{s_{l}^{n}}(t)\right)\right)\bigg\}\nonumber \\
 & +\sum_{l\in\mathbb{A}_{\alpha}^{c}}\bigg\{ f\left(X^{s_{l}^{n}}(t)+p_{\{j\,|\pi(j)\leq\pi(i)\}}\left(X^{s_{l+1}^{n}}(t)-X^{s_{l}^{n}}(t)\right)\right)\nonumber \\
 & \quad\quad\quad-f\left(X^{s_{l}^{n}}(t)+p_{\{j\,|\pi(j)<\pi(i)\}}\left(X^{s_{l+1}^{n}}(t)-X^{s_{l}^{n}}(t)\right)\right)\bigg\},\label{eq:D_i}
\end{align}
where $\mathbb{A}_{\alpha}=\{l\in\mathbb{N}_{0}:\mathcal{A}_{\alpha}\cap(s_{l}^{n},s_{l+1}^{n}]\neq\emptyset\}$ and $\mathbb{A}_{\alpha}^{c}=\mathbb{N}_{0}\setminus\mathbb{A}_{\alpha}$.
The first sum at the right-hand side of Eq.~(\ref{eq:D_i}) converges
a.s. for $n\to\infty$ to 
\begin{equation}
\sum_{s\in\mathcal{A}_{\alpha}}\bigg\{ f\left(X(s-)+p_{\{j\,|\pi(j)\leq\pi(i)\}}\big(\Delta X(s)\big)\right)-f\left(X(s-)+p_{\{j\,|\pi(j)<\pi(i)\}}\big(\Delta X(s)\big)\right)\bigg\}.\label{eq:sumY}
\end{equation}
Using a Taylor expansion and the same arguments as in the proof of the
classical Itô's formula, one can show that the second sum of the right-hand
side of Eq.~(\ref{eq:D_i}) converges in probability for $n\to\infty$
to
\begin{align}
I_{i}(t)+\frac{1}{2}H_{ii}(t)+\sum_{\pi(j)<\pi(i)}H_{ij}(t)-\sum_{s\in\mathcal{A}_{\alpha}}\bigg\{ f_{i}\left(X(s-)\right)\Delta X_{i}(s)+\frac{1}{2}f_{ii}\left(X(s-)\right)\left(\Delta X_{i}(s)\right)^{2}\nonumber \\
+\sum_{\underset{\pi(j)<\pi(i)}{j=1}}^{d}f_{ij}\left(X(s-)\right)\Delta X_{i}(s)\Delta X_{j}(s)\bigg\},\label{eq:jumpsdelta}
\end{align}
where $H_{ij}=\int_{0}^{\cdot}f_{ij}(X(s-))d[X_{i},X_{j}](s)$. The sum of the Eqs.~(\ref{eq:sumY}) and (\ref{eq:jumpsdelta}) is
\begin{align}
 & I_{i}(t)+\frac{1}{2}H_{ii}(t)+\sum_{\pi(j)<\pi(i)}H_{ij}(t)+\label{eq:IH}\\
 & \sum_{s\in\mathcal{A}_{\alpha}}\bigg\{ f\left(X(s-)+p_{\{j\,|\pi(j)\leq\pi(i)\}}(\Delta X(s))\right)-f\left(X(s-)+p_{\{j\,|\pi(j)<\pi(i)\}}(\Delta X(s))\right)\nonumber \\
 & \quad\quad\quad-f_{i}\left(X(s-)\right)\Delta X_{i}(s)\bigg\}\label{eq:sum A1}\\
 & -\sum_{s\in\mathcal{A}_{\alpha}}\frac{1}{2}f_{ii}\left(X(s-)\right)\left(\Delta X_{i}(s)\right)^{2}\label{eq:sum A2}\\
 & -\sum_{s\in\mathcal{A}_{\alpha}}\sum_{\underset{\pi(j)<\pi(i)}{j=1}}^{d}f_{ij}\left(X(s-)\right)\Delta X_{i}(s)\Delta X_{j}(s).\label{eq:A2b}
\end{align}
Since $X$ is a semimartingale, and because of  Lemma \ref{lem:S^A_abs_con}, we can see that the
sums (\ref{eq:sum A1}), (\ref{eq:sum A2}) and (\ref{eq:A2b}) are
absolutely convergent for $\alpha\to0$ so that (\ref{eq:IH}-\ref{eq:A2b})
converge for $\alpha\to0$ to $\delta^{\text{ISU},\pi}(F,X)(t)$,
using that
\[
I_{ij}=H_{ij}-\sum_{0\leq s\leq\cdot}f_{ij}\left(X(s-)\right)\Delta X_{i}(s)\Delta X_{j}(s).
\]
 By Theorem \ref{thm:IASU}  we get $\delta^{\text{ASU},\gamma_{n}}(F,X)(t)\overset{p}{\to}\delta^{\text{IASU},\pi}(F,X)(t)$ for $n\to\infty$.
\end{proof}

\subsection{\label{subsec:Stability}Stability}

In this section, we use the notation of \citet{christiansen2022decomposition}.
For $i=1,2$, let $\tau_{i}:[0,\infty)\to[0,\infty)$ with $\tau_{i}(t)\leq t$
for all $t\geq0$. The function 
\[
\tau(t)=\big(\tau_{1}(t),\tau_{2}(t)\big)
\]
is called a \emph{delay}. A delay is called \emph{phased} if there
is an unbounded partition $\{0=s_{0}<s_{1}<...\}$ of $[0,\infty)$
such that on each interval $(s_{l},s_{l+1}]$, at most one component
of $\tau$ is nonconstant. Let $(\tau^{n})_{n\in\mathbb{N}}$ be a
\emph{refining sequence of delays that increase to identity} (rsdii),
i.e., 
\[
\tau_{i}^{n}\big([0,t]\big)\subset\tau_{i}^{n+1}\big([0,t]\big),\;n\in\mathbb{N},\quad\text{ and }\overline{\bigcup_{n\in\mathbb{N}}\tau_{i}^{n}\big([0,t]\big)}=[0,t],\quad i=1,2.
\]
Let $\mathcal{T}$ be a set containing at least one phased rsdii.
Let $X=(X_{1},X_{2})$ be a semimartingale, and define 
\[
X\diamond\tau:=(X_{1}\circ\tau_{1},X_{2}\circ\tau_{2}),\quad\tau\in\mathcal{T}.
\]
Let 
\[
\mathbb{X}=\{X\diamond\tau\,:\,\tau\in\mathcal{T}\}\cup\{X\}.
\]
Let $\mathbb{D}_{0}$ be the set of càdlàg processes starting in zero
and let $\varrho:\mathbb{X}\to\mathbb{D}_{0}$. A mapping $\delta:\mathbb{X}\to\mathbb{D}_{0}^{2}$
is called \emph{decomposition scheme of $\varrho$.} The mapping $\delta$ assigns to
each $Y\in\mathcal{\mathbb{X}}$ a decomposition of $\varrho(Y)$.
The ISU decomposition scheme is abbreviated $\delta^{ISU}$. A decomposition
scheme is called \emph{stable} \emph{at} $X$ if 
\[
\delta(X\diamond\tau^{n})(t-)\overset{p}{\to}\delta(X)(t-),\quad n\to\infty,
\]
at each  $t >0$  for all rsdii $(\tau^{n})_{n\in\mathbb{N}}\subset\mathcal{T}$.
\begin{prop}
Assume that $X=(X_{1},X_{2})$ with $X_{1}=X_{2}=B$ for a Brownian
motion $B$. Let $\varrho(Y)=Y_{1}Y_{2}$ be a simple product. Then,
there is a set $\mathcal{T}$ of continuous phased rsdii such that
the ISU decomposition $\delta^{ISU}$ of $\varrho$ is not stable
at $X$. 
\end{prop}

\begin{proof}
Suppose that $\mathcal{T}$ contains a continuous phased rsdii $(\tau^{n})=(\tau_{1}^{n},\tau_{2}^{n}),n\in\mathbb{N}$,
with $\tau_{1}^{n}\leq\tau_{2}^{n},\,n\in\mathbb{N}$. For a partition
$(a_{l,i}^{n},b_{l,i}^{n}],\,l\in\mathbb{N}_{0},\,i=1,2$ of $[0,\infty)$
such that $(\tau_{j}^{n})_{j\neq i}$ is constant on $(a_{l,i}^{n},b_{l,i}^{n}]$,
let $\tau_{1}^{n}(a_{l,2}^{n})=\tau_{2}^{n}(a_{l,2}^{n}),\,n\in\mathbb{N},\,l\in\mathbb{N}_{0}$.
In addition, let $\mathcal{T}$ also contain $(\tilde{\tau}^{n})_{n\in\mathbb{N}}=((\tau_{2}^{n},\tau_{1}^{n}))_{n\in\mathbb{N}}$.
Since $\tau_{2}^{n}(a_{l,1}^{n})=\tau_{2}^{n}(b_{l,1}^{n})=\tau_{1}^{n}(b_{l,1}^{n})$
and by the multidimensional Taylor theorem, 
\begin{align*}
\delta_{1}^{ISU}(X\diamond\tau^{n})(t)= & \sum_{l}\bigg(\varrho\big((X\diamond\tau^{n})^{b_{l,1}^{n}\wedge t}\big)-\varrho\big((X\diamond\tau^{n})^{a_{l,1}^{n}\wedge t}\big)\bigg)\\
= & \sum_{l}\varrho_{1}\big((X\diamond\tau^{n})^{a_{l,1}^{n}\wedge t}\big)\bigg(X_{1}\big(\tau_{1}^{n}(b_{l,1}^{n}\wedge t)\big)-X_{1}\big(\tau_{1}^{n}(a_{l,1}^{n}\wedge t)\bigr)\bigg).\\
\end{align*}
By the definitions of $X_{1}$, $X_{2}$ and $\rho$, 
\begin{align*}
\delta_{1}^{ISU}(X\diamond\tau^{n})(t)= & \sum_{l}B\big(\tau_{2}^{n}(a_{l,1}^{n}\wedge t)\big)\bigg(B\big(\tau_{1}^{n}(b_{l,1}^{n}\wedge t)\big)-B\big(\tau_{1}^{n}(a_{l,1}^{n}\wedge t)\bigr)\bigg)\\
= & \sum_{l}B\big(\tau_{1}^{n}(b_{l,1}^{n}\wedge t)\big)\bigg(B\big(\tau_{1}^{n}(b_{l,1}^{n}\wedge t)\big)-B\big(\tau_{1}^{n}(a_{l,1}^{n}\wedge t)\bigr)\bigg)\\
= & \sum_{l}B(t_{l})\bigl(B(t_{l}\wedge t)-B(t_{l-1}\wedge t)\bigr)\\
= & 2\sum_{l}\frac{\bigl(B(t_{l})+B(t_{l-1})\bigr)}{2}\bigl(B(t_{l}\wedge t)-B(t_{l-1}\wedge t)\bigr)\\
 & -\sum_{l}B(t_{l-1})\bigl(B(t_{l}\wedge t)-B(t_{l-1}\wedge t)\bigr)
\end{align*}
for $t_{l}^{n}:=\tau_{1}^{n}(b_{l,1}^{n})=\tau_{1}^{n}(a_{l+1,1}^{n})=\tau_{2}^{n}(b_{l-1,2}^{n})=\tau_{2}^{n}(a_{l,2}^{n})$.
Let $\int_{0}^{t}B_{s}\circ dB_{s}$ denote the Stratonovich integral
and $\int_{0}^{t}B_{s}dB_{s}$ the Itô integral. It holds that 
\begin{align*}
\delta_{1}^{ISU}(X\diamond\tau^{n})(t) & \overset{p}{\rightarrow}2\int_{0}^{t}B_{s}\circ dB_{s}-\int_{0}^{t}B_{s}dB_{s}\\
 & =\frac{1}{2}B_{t}^{2}+\frac{1}{2}t
\end{align*}
for $n\rightarrow\infty$. By the same arguments,
\begin{align*}
\delta_{1}^{ISU}(X\diamond\tilde{\tau}^{n})(t) & =\sum_{l}\bigg(\varrho\big((X\diamond\tilde{\tau}^{n})^{b_{l,2}^{n}\wedge t}\big)-\varrho\big((X\diamond\tilde{\tau}^{n})^{a_{l,2}^{n}\wedge t}\big)\bigg)\\
 & =\sum_{l}\varrho_{1}\big((X\diamond\tilde{\tau}^{n})^{a_{l,2}^{n}\wedge t}\big)\bigg(X_{2}\big(\tau_{2}^{n}(b_{l,2}^{n}\wedge t)\big)-X_{2}\big(\tau_{2}^{n}(a_{l,2}^{n}\wedge t)\bigr)\bigg)\\
 & =\sum_{l}B\big(\tau_{1}^{n}(a_{l,2}^{n}\wedge t)\big)\bigg(B\big(\tau_{2}^{n}(b_{l,2}^{n}\wedge t)\big)-B\big(\tau_{2}^{n}(a_{l,2}^{n}\wedge t)\big)\bigg)\\
 & =\sum_{l}B\big(\tau_{2}^{n}(a_{l,2}^{n}\wedge t)\big)\bigg(B\big(\tau_{2}^{n}(b_{l,2}^{n}\wedge t)\big)-B\big(\tau_{2}^{n}(a_{l,2}^{n}\wedge t)\big)\bigg)\\
 & =\sum_{l}B(t_{l})\bigl(B(t_{l+1}\wedge t)-B(t_{l}\wedge t)\bigr)\\
 & \overset{p}{\rightarrow}\int_{0}^{t}B_{s}dB_{s}\\
 & =\frac{1}{2}B_{t}^{2}-\frac{1}{2}t
\end{align*}
for $n\rightarrow\infty$. Therefore, 
\[
\text{plim}_{n\rightarrow\infty}\delta_{i}^{ISU}(X\diamond\tau^{n})(t)\neq\text{plim}_{n\rightarrow\infty}\delta_{i}^{ISU}(X\diamond\tilde{\tau}^{n})(t),\quad i=1,2,
\]
for $t>0$, and hence, the ISU decomposition of $\varrho(X)$ cannot
be stable at $X$. 
\end{proof}

\end{document}